\documentclass[12pt,reqno]{amsart}
\usepackage{amsmath}
\usepackage{amsfonts}
\usepackage{amsthm}
\usepackage{amstext}
\usepackage{enumerate}
\usepackage[latin1]{inputenc}

\usepackage{graphicx}
\usepackage{verbatim}

     \newcommand{\Ran}{{\operatorname{Ran}}}

     \newcommand{\N}{{\mathbb{N}}}

     \newcommand{\R}{{\mathbb{R}}}
     \newcommand{\Z}{{\mathbb{Z}}}
     \newcommand{\C}{{\mathbb{C}}}

\newcommand{\e}{{\rm e}}

\renewcommand{\i}{{\rm i}}
\renewcommand{\d}{{\rm d}}

\newcommand{\rad}{{\rm rad}}

\renewcommand{\Im}{{\rm Im}\,}

\newcommand\inp[2][]{#1 \langle #2#1\rangle}
\newcommand\parb[2][]{#1 \big ( #2#1\big )}
\newcommand\parbb[2][]{#1 \Big ( #2#1\Big )}

\renewcommand{\exp}{{\rm exp}}
\newcommand{\mand}{\text{ and }}
\newcommand{\mfor}{\text{ for }}
\newcommand{\mforall}{\text{ for all }}

\newcommand{\vB}{{\mathcal B}}

\newcommand{\vG}{{\mathcal G}}
\newcommand{\vI}{{\mathcal I}}

\newcommand{\vH}{{\mathcal H}}

\newcommand{\vS}{{\mathcal S}}

     \theoremstyle{plain}%default
     \newtheorem{thm}{Theorem}[section]
     \newtheorem{prop}[thm]{Proposition}
     \newtheorem{lemma}[thm]{Lemma}
      \newtheorem{cor}[thm]{Corollary}
     \theoremstyle{definition}
     
     \newtheorem{defn}[thm]{Definition}

     \newtheorem{remark}[thm]{Remark}
     \newtheorem{remarks}[thm]{Remarks}

\newtheorem*{remarks*}{Remarks}
\newtheorem*{remark*}{Remark}

     \numberwithin{equation}{section}

\begin{document}

\title[Zero]{Analyticity estimates for the Navier-Stokes equations}

\author{
I. Herbst    \\
{\tiny Department of Mathematics} \\
{\tiny University of Virginia}  \\
{\tiny  Charlottesville, VA 22904 U.S.A.} \\
{\tiny iwh@virginia.edu} \\
\\[8mm]
E. Skibsted   \\
{\tiny Institut for  Matematiske Fag} \\
{\tiny Aarhus Universitet} \\
{\tiny 8000 Aarhus C, Denmark} \\
{\tiny skibsted@imf.au.dk}
}

%\thanks{The authors are  partially supported by  ..}
\date{\today}


\begin{abstract}
We study spatial analyticity properties of solutions of the Navier-Stokes equation and obtain new growth rate estimates for the analyticity radius.  We also study stability properties of strong global solutions of the Navier-Stokes equation with data in $H^r, \ r \geq 1/2$ and prove a stability result for the analyticity radius.
\end{abstract}

\maketitle

\setcounter{tocdepth}{1}

{\small \tableofcontents}

\section{Introduction}

We look at the following system of equations in the variables  $(t,x)\in
[0,\infty[\times \R^3$. The unknown $u$ specifies at each argument a
velocity $u=u(t,x)\in \R^3$. The unknown $p$ specifies at each
argument a pressure  $p=p(t,x)\in \R$.
\begin{equation}\label{eq:equamotiona}
\begin{cases}
\tfrac{\partial }{\partial t}u+(u\cdot \nabla)u +\nabla p=\triangle u\\
\nabla \cdot u=0\\
u(t=0)=u_0
\end{cases}\;.
\end{equation}

 We eliminate the pressure in the standard way using the
  Leray
projection
 $P$. It is an orthogonal projection on  $L^2(\R^3,\R^3)$ which
is fibered in Fourier space, i.e. $\widehat {(Pf)}(\xi)=P(\xi) \hat f
(\xi)$, and it is given according to the following recipe:
$P(\xi) =I-|\hat \xi\rangle \langle \hat\xi |
,\;\hat\xi
:=\xi/|\xi|$.

\begin{equation}\label{eq:equamotionBA}
\begin{cases}
\big (\tfrac{\partial }{\partial t}u+P(u\cdot \nabla)u -\triangle
u\big )(t,\cdot)=0\\
u(t):=u(t,\cdot)\in\Ran P
\end{cases}\;.
\end{equation}

We introduce a notion of strong global solution to
(\ref{eq:equamotionBA}) in terms of $A:=\sqrt{-\triangle}$.  In the following the Sobolev space $H^r$ is the Hilbert space with norm $\|f\|_{H^r} = \|\langle A \rangle^r f\|_{L^2}$ where $\langle x \rangle = \sqrt{1 + |x|^2}$.

\begin{defn}\label{defn:strong_a} Let  $r\geq 1/2$. The set
 $\vG_{r}$  is the  set of $u\in C([0,\infty[,(H^r)^3)$ satisfying:
  \begin{enumerate}[(1)]

  \item \label{item:20a} $u(t)\in P(H^r)^3$ $\mforall t\geq0$,

  \item \label{item:21a} The expression
$A^{5/4}u(t)$ defines an element in $C(]0,\infty[,(L^2)^3)$ and
\begin{equation*}
  \lim_{t\to 0}t^{3/8}\|A^{5/4}u(t)\|_{(L^2)^3} = 0,
\end{equation*}

\item \label{item:22a} $u\in C^1(]0,\infty[,
  \vS'(\R^3) )$ and
\begin{equation*}
   \tfrac{\d}{\d t}u=-A^2u-P(u\cdot \nabla)u;\,t>0.
  \end{equation*}
  \end{enumerate}
Here the differentiability in $t$ is meant in the weak* topology and the equation in (3) is meant in the sense of distributions.
 We
  refer to any $u\in \vG_{r}$ as  a {\it  strong global solution} to the
  problem
  (\ref{eq:equamotionBA}).
\end{defn}

\subsection{Discussion of uniform
 real  analyticity} \label{Discussion of uniform
 real  analyticity}

 Fix $r\in \R $ and $f\in H^r$. We say that $f$ is {\it uniformly $
 H^r$ real analytic} if there exists  $a>0$ such that the function $\R^3\ni
 \eta\to \e^{\i \eta\cdot p}f=f(\cdot+\eta)\in  H^r$, $p:=-\i \nabla$,  extends to an
 analytic function $\tilde f$
 on $\{|\Im \eta |<a\}$ and that
 \begin{equation}
   \label{eq:37}
   \sup_{|\Im
\eta|<a}\|\tilde f(\eta)\|_{H^r}<\infty.
 \end{equation}  We define correspondingly the {\it analyticity radius } of $f$ as
 \begin{equation}
   \label{eq:42}
   \rad (f)=\sup \{a>0|\text{ the property } (\ref{eq:37})\text{ holds}\}.
 \end{equation} If $f$ is not uniformly $
 H^r$ real analytic we put $\rad (f)=0$.

We note that the notions of  uniform real analyticity and corresponding analyticity
radius  are  independent of $r$, and that in fact
\begin{equation}
  \label{eq:43}
  \rad(f)=\sup \{a\geq 0|\,
  \e^{ aA}f\in H^r\}=\sup \{a\geq 0|\, \e^{ aA}f\in L^2\}.
\end{equation}
Moreover if $\rad(f)>0$ then  by the Sobolev embedding theorem, $H^{s}\subseteq L^\infty$ for  $s>3/2$, the function
 $\R^3\ni
 x\to f(x)\in  \C$ extends to the analytic function $\breve f$
 on $\{|\Im \eta |<\rad(f)\}$ given by $\breve f(\eta)=\tilde
 f(\eta)(0)$. Conversely suppose a given function
 $\R^3\ni
 x\to f(x)\in  \C$ extends to an analytic function $\breve f$ on $\{|\Im \eta |<b\}$
 and that $\tilde
 f(\eta) :=\breve f(\cdot+\eta)$ obeys (\ref{eq:37}) for all $a<b$
 then $\tilde
 f$ is an analytic $H^r$-valued function and $\rad(f)\geq b$.

If $f\in \dot H^r$ for some $r\in ]-\infty,3/2[$ (see Section
\ref{Integral equation} for the definition of homogeneous Sobolev
spaces) one can introduce  similar
notions of  uniform real analyticity and corresponding analyticity
radius (by using (\ref{eq:37}) with $H^r\to \dot
H^r$ and (\ref{eq:42}), respectively). If $f\in \dot H^r$ has a positive analyticity radius then $f = f_1 +f_2$ where
$f_1$ has an entire analytic continuation and $f_2$ is uniformly $H^s$ real analytic for any $s$.  Note however that the concept of uniform $\dot H^r$ analyticity is dependent on $r$.

In any case
$a>0$ will be a lower bound of the analyticity radius of  $f\in H^r$
or $f\in \dot H^r$ if $\e^{ aA}f\in H^r$ or  $\e^{ aA}f\in \dot
H^r$, respectively.

For $f=(f_1,f_2,f_3)\in (H^r)^3$ we define $\rad(f)=\min_j(\rad(f_j))$.

\subsection{Results on real analyticity of solutions} \label{Results on analyticity}
It  is a basic fact that for all  $u\in \vG_{r}$,  $r\geq1/2$,
\begin{equation}\label{eq:79}
  \rad(u(t))>0 \mforall t>0.
\end{equation}
In fact for any such $u$ there exists $\lambda >0$ such that
\begin{equation}\label{eq:79kkk}
  \rad (u(t))\geq \lambda \sqrt t \mforall t>0;
\end{equation}
 we note that \eqref{eq:79kkk} is a consequence of Corollary
\ref{cor:sobol-analyt-bounds}, Lemma \ref{lemma:energyre}  and
Proposition \ref{prop:kato big data anal}.

The main subject of this paper is the study of lower
bounds of the quantity to the left in (\ref{eq:79}) (and the analogous
question for solutions taking values in homogeneous Sobolev spaces). The study of analyticity of solutions of the Navier-Stokes equations originated with Foias and Temam in [FT] where they studied analyticity of periodic solutions in space and time (see also [FMRT]).  Later Gruji$\check {\rm c}$ and Kukavica [GK] studied space analyticity of the Navier-Stokes equations in $R^3$.   Since then many authors have proven analyticity results.  We mention the book by Lemari\'e-Rieusset [Le] where some results and references can be found.

There are two regimes  to study, the small and the large time regimes.
One of our  main results  on large time analyticity bounds is the following (a
combination of Theorem \ref{thm:opt} \ref{item:8} and Lemma \ref{lemma:energyre}):

\begin{thm}\label{thm:opta}
  Suppose $r\geq1/2$ and $u\in \vG_{r}$.
Suppose that for some $\sigma \geq0$  the
  following bound holds
\begin{equation}
    \label{eq:69ba}
    \|u(t)\|_{(L^2)^3}=O(t^{-\sigma/2})\mfor t\to \infty.
  \end{equation}

Let $0\leq \tilde
   \epsilon<\epsilon\leq 1$ be given. Then there exist constants
   $t_0>1$ and $C>0$  such that
   \begin{equation}
     \label{eq:68ba}
     \|\exp\parbb {\sqrt{(1-\epsilon)(2\sigma+1)}\sqrt{t\ln t}A} u(t)\|_{(L^2)^3}\leq C t^{1/4-\tilde \epsilon (2\sigma+1)/4}\mforall
   t\geq t_0.
   \end{equation}

In particular
 \begin{equation}
    \label{eq:46wwa}
 \liminf_{t\to \infty}\tfrac{\rad(u(t))}{\sqrt{t\ln t}}\geq \sqrt{2\sigma+1}.
  \end{equation}

\end{thm}
\begin{remarks*}
  \begin{enumerate}[1)]
  \item \label{item:19} For any $u\in \vG_{r}$ the quantity $ \|u(t)\|_{(L^2)^3}=o(t^0)$
  as $t\to \infty$. In particular (\ref{eq:69ba}) is valid for $\sigma
  =0$. Under some further conditions (partly generic, involving the condition $\int u_{0i}u_{0j}dx \neq c \delta_{ij}, u_0 = u(0)$) it is shown in \cite{Sc1} that  for some $C > 0, \ C^{-1}\langle t \rangle^{-5/4} \leq \|u(t)\|_{(L^2)^3} \leq C\langle t \rangle^{-5/4}$. References to many further works on the $L^2$ decay rate can be found in \cite{Sc1}  .

\item \label{item:23} Our method  of proof works more generally, in
  particular for the class of compressible flows given by taking $P=I$
  in (\ref{eq:equamotionBA}). In this setting we construct an example
  for which (\ref{eq:69ba})--(\ref{eq:46wwa}) hold with $\sigma=5/2$ and for
  which {\it all} of the bounds (\ref{eq:69ba})--(\ref{eq:46wwa}) with
  $\sigma=5/2$ are
  {\it optimal}.
If $P$ is the  Leray
projection in (\ref{eq:equamotionBA}) we do not know if
 the bounds (\ref{eq:68ba}) and (\ref{eq:46wwa}) are optimal  under
 the decay condition  (\ref{eq:69ba}).
  \end{enumerate}
\end{remarks*}
As for the small time regime, we have less complete knowledge. It  is a basic fact that for all  $u\in
\vG_{r}$,  $r\geq1/2$,
\begin{equation}
  \label{eq:32Haa}
  \lim_{t\to 0}\tfrac {\rad (u(t))}
{\sqrt t}=\infty.
\end{equation} This result is valid for any strong solution to
(\ref{eq:equamotionBA}) defined on an  interval $I=]0,T]$ only (i.e.
$I=]0,\infty[$ is  not needed here). This
class, $\vS_{r,I}$, is introduced in Definition \ref{defns:strong}
in a similar way as the set of strong global solutions to
(\ref{eq:equamotionBA}) is introduced  in
Definition \ref{defn:strong_a}. For any $u_0\in (H^r)^3$, $r\geq1/2$,
there exists a unique  strong solution $u\in \vS_{r,I}$ with
$u(0)=u_0$ to the equations (\ref{eq:equamotionBA}) provided that
$T=|I|$ is small enough (alternatively, a
unique small time solution to the initial value problem (\ref{eq:equamotiona})).

One of our  main results  on small  time analyticity bounds is the
following (cf. Corollary \ref{cor:simple}):
\begin{thm}\label{thm:simpleaa} Suppose $u_0\in (H^r)^3$ for some $r\in
  ]1/2,3/2[$. Let $u$ be the unique small time ($T=|I|$ small)  strong solution $u\in
  \vS_{r,I}$  with
$u(0)=u_0$. Let $\epsilon\in
  ]0,2r-1]$. Then there exist constants  $t_0=t_0\big (\epsilon, r,
  \|\inp{A}^ru_0\|\big )\in ]0,T]$ and $C=C\big (\epsilon, r,
  \|\inp{A}^ru_0\|\big )>0$   such
that
  \begin{equation}
    \label{eq:29aa}
    \|\e^{\sqrt{2r-1-\epsilon}\ \sqrt{t|\ln t|}A}u(t)\|_{
    (H^{r})^3}\leq Ct^{1/4+\epsilon/4-r/2}\mforall t\in]0,t_0].
  \end{equation} In particular
  \begin{equation}
    \label{eq:46aa}
 \liminf_{t\to 0}\tfrac{\rad(u(t))}{\sqrt{t|\ln t|}}\geq \sqrt{2r-1}.
  \end{equation}
\end{thm}
\begin {remark*} There are several natural questions connected
with these results:  Are
 the bounds (\ref{eq:29aa}) and (\ref{eq:46aa}) optimal for  $r\in
  ]1/2,3/2[$? Are there better bounds than those
  deducible  from
  Theorem \ref{thm:simpleaa}
  if $r>3/2$? Can
  (\ref{eq:32Haa}) be improved for $r=1/2$? (But in this connection see the discussion of an example in
  Subsection 3.4.)
\end{remark*}

\subsection{Results on stability of the analyticity radius}\label{Results on openess for global solvability}
\begin{defn}
  \label{def:global-stabilitya}For   $r\geq 1/2$ we denote by
\begin{equation}
  \label{eq:78a}
  \vI_r=\{u_0\in P(H^r)^3| \,\exists u\in \vG_r: \,u(0)=u_0\},
\end{equation} and we endow $ \vI_r$ with the topology from the space $P(H^r)^3$.
\end{defn}

%We note that for all $u_0\in \vI_r$   the corresponding strong global solution
%  $u$ is unique.

Our main result  on stability of the region of analyticity of global solutions  to
(\ref{eq:equamotionBA}) is the following (from Theorem \ref{thm:global-stability} and Lemma \ref{lemma:energyre}).
\begin{thm}\label{thm:global-stabilitya}  For all $r\geq 1/2$ the set
  $\vI_r$ is open in $ P(H^r)^3$. Given $u_0 \in \vI_{1/2}$, if
  $\lambda>0$ is given  so that the corresponding solution u(t)
  satisfies
  \begin{equation}
    \label{eq:81pppppppp}
    A^{1/2}\e^{\lambda \sqrt \cdot A}u(\cdot) \in C([0,\infty[, L^2),
  \end{equation}
then
  there is a $\delta_0 > 0$ so that if $\delta \leq  \delta_0$ and  $v_0
  \in P(H^{1/2})^3$ with $\|A^{1/2}(v_0 - u_0)\| \leq  \delta$  the
  solution $v$ with initial data $v_0$ is in $\vG_{1/2}$ and satisfies
\begin{subequations}
  \begin{align}
  \label{eqn:radiusstabilityq}
  \|A^{1/2}\e^{\lambda \sqrt tA}(v(t)-u(t))\| &\leq  K_1\delta,\\
  \label{eqn:radiusstability2q}
  t^{3/8}\|A^{5/4}\e^{\lambda \sqrt tA}(v(t)-u(t))\|&\leq K_2\delta.
  \end{align}
If  $\|v_0 - u_0\|_{H^{1/2}} \leq \delta$ it follows in addition
that
\begin{equation}
  \label{eqn:radiusstability322}
  \inp{t}^{-1/4}\|\e^{\lambda \sqrt tA}(v(t)-u(t))\|\leq K_3\delta.
\end{equation}
\end{subequations}
In \eqref{eqn:radiusstabilityq}--\eqref{eqn:radiusstability322}  the
 constants $K_1,K_2,K_3>0$ depend on $\lambda$, $u$, and $\delta_0$ but not on
$\delta$, and all bounds are uniform in $t>0$.
  \end{thm}

We note that the fact that $\vI_r$ is open is a known result.
References will be given in Subsection \ref{Completing the proof}. We
also note that indeed for any $u_0 \in \vI_{1/2}$ the condition
\eqref{eq:81pppppppp} holds for some $\lambda>0$,
cf. \eqref{eq:79kkk}. We
apply  Theorem \ref{thm:global-stabilitya} (and some other results of
this paper) to establish  a new stability
result for the $L^2$ norm. This result is presented in Subsection
\ref{L^2 decay stability}.
\begin {remark*} There is a   natural question connected
with Theorem \ref{thm:global-stabilitya}:  Is the analyticity radius
lower semicontinuous in the $ H^{1/2}$ topology? More precisely one
may conjecture  that for any fixed $u_0 \in \vI_{1/2}$  and $t>0$
\begin{equation} \label{eq:semi-continuity}
  \liminf\,\rad (v(t))\geq \rad
  (u(t)) \text{ in the limit } \|v_0 - u_0\|_{H^{1/2}}\to 0?
\end{equation} For partial  results in this direction see
Proposition \ref{thm:global-stability_0} and Corollary \ref{cor:semicontinuity}.
\end{remark*}

 We shall use the standard notation
 $\inp{\lambda}:=(1+|\lambda|^2)^{1/2}$ for any real $\lambda$. For
 any given interval $J$ and Hilbert space $\vH$ the notation
 $BC(J,\vH)$ refers to the set of all bounded continuous functions $v:
 J\to \vH$.

\section{Integral equation} \label{Integral equation}
We look at the following (generalized) system of equations in the variables  $(t,x)\in
[0,\infty[\times \R^3$. The unknown $u$ specifies at each argument a
velocity $u=u(t,x)\in \R^3$.  The quantity $M$  is a fixed real $3\times
3$--matrix. Corresponding to (\ref{eq:equamotiona}) $M=I$. The
quantity $P$ is an orthogonal projection on  $L^2(\R^3,\R^3)$ which
is fibered in Fourier space, i.e. $\widehat {(Pf)}(\xi)=P(\xi) \hat f
(\xi)$.  Corresponding to (\ref{eq:equamotiona}) $P$ is the Leray
projection given by
$P(\xi) =I-|\hat \xi\rangle \langle \hat\xi |,\;\hat\xi
:=\xi/|\xi|$, but for the most part we will not assume this.
\begin{equation}\label{eq:equamotion}
\begin{cases}
\big (\tfrac{\partial }{\partial t}u+(Mu\cdot \nabla)u -\triangle
u\big )(t,\cdot)\in\Ran\parb {I-P}\mfor t>0\\
u(t):=u(t,\cdot)\in\Ran P\mfor t\geq 0\\
u(0)=u_0
\end{cases}\;.
\end{equation}
 Similarly the operator $A:=\sqrt{-\triangle}$ on  $L^2(\R^3,\R^3)$
 is fibered in Fourier space as  $\widehat {(Af)}(\xi)=|\xi|\hat f (\xi)$.  Upon multiplying
the first equation by $P$ and integrating
we obtain (formally)
\begin{equation}
  \label{eq:1}
  u(t)=\e^{-tA^2}u_0-\int _0^t\e^{-(t-s)A^2}P(Mu(s)\cdot \nabla)u(s)\ \d s.
\end{equation}

Conversely,  notice  (formally) that a solution to (\ref{eq:1}) with $u_0\in \Ran P$
obeys (\ref{eq:equamotion}). In the  bulk of this paper we shall study
(\ref{eq:1}) without imposing the condition  $u_0\in \Ran P$. See
though Section \ref{Openness of the set of global solutions with
  respect to data} for an exception. In fact  in Subsection
\ref{Global solutions} we shall  study (under some
conditions)   the
relationship between (\ref{eq:equamotion}) and (\ref{eq:1}). For the
bulk of this paper this  relationship is minor although traces are used
already
in Sections \ref{Differential inequalities for small
  global solutions in dot 1/2  and  1/2} and \ref{Optimal rate of growth of analyticity
  radius }.  The reader might  prefer to read the present section and Subsection
\ref{Global solutions} before proceeding to Section \ref{Analyticity bounds
  for small times}.

In this section we consider the Cauchy problem in the form (\ref{eq:1}) using norms based essentially on Sobolev spaces.  Although this material is well known (see [FK1], [FK2], and for example [KP], [Le], [Pl]) we give a self-contained account so that we can use the specific results and methods in our analysis of the spatial analyticity of solutions of (\ref{eq:1}) in the sections following.

Part of our motivation for studying equations more general than (\ref{eq:equamotionBA}) is that
such a study emphasizes what we actually use in our analysis.  In particular, besides the case where $M = I$ and $P$ is the Leray projection, we will consider the vector Burgers' equation, as an example, where $M = I$ and $P = I$ (see Subsections \ref{Example1} and \ref{Example2}). The latter equation has been studied in [KL], [JS], [Ga], and elsewhere.

We define for  any $r\in ]-\infty,3/2[$ the homogeneous Sobolev space $\dot H^r$ to
be the set of $f\in \vS'$ such that the Fourier transform $\hat f$ is
a measurable function and $|\xi|^r\hat f(\xi)\in L^2(\R_{\xi}^3)$. The corresponding
norm is $\|f\|_{\dot H^r}=\|A^rf\|$ where here and henceforth   $\|\cdot\|$ refers to the
$L^2$-norm. For simplicity we shall use the same notation for vectors
$f\in (L^2)^3=L^2\oplus L^2\oplus L^2$, viz. $\|f\|=\sqrt
{\|f_1\|^2+\|f_2\|^2+\|f_3\|^2}$ for $f=(f_1,f_2,f_3)\in (L^2)^3$.

Let $I$  be an interval of the form $I=]0,T]$ (if  $T$ is finite)
or  $I=]0,T[$ (if $T=\infty$). The closure $I\cup\{0\}$  will be denoted
$\bar I$. Let  $\zeta:I\to \R$ and $\theta:I\to [0, \infty[$ be  given continuous functions.  Let
$s_1,s_2\in [0,3/2[$ be given. We shall consider the class of
functions $I\ni t\to v(t)\in (\dot H^{s_2})^3$ for which the expression
$\e^{-\zeta(t)}t^{s_1}\e^{\theta(t)A}A^{s_2}v(t)$ defines an element in $BC(I,(L^2)^3)$. The
set of such functions,  denoted by
$\vB_{\zeta,\theta,I,s_1,s_2}$, is a Banach space with the norm
\begin{equation}
  \label{eq:3}
  \|v\|_{\zeta, \theta,I,s_1,s_2}:=\sup_{t\in I}\e^{-\zeta(t)}t^{s_1}\|\e^{\theta(t)A}A^{s_2}v(t)\|.
\end{equation}

In this section we discuss the case $\zeta=0$ and $\theta=0$ only which (upon choosing $s_1$ and
$s_2$ suitably) corresponds  to part of the pioneering
work \cite {KF1, KF2}. Consequently we omit throughout this section
the subscripts $\zeta$ and $\theta$ in the above notation.

We recall the
following class of Sobolev bounds (cf. \cite[(IX.19)]{RS}):

\begin{lemma}
  \label{lemma:sobolev} For all $r\in ]0,3/2[$, $f\in \dot H^r$ and
  all $g\in \dot H^{3/2-r}$ the product $fg\in L^2$, and  there exists a constant
  $C=C(r)>0$ such that
  \begin{equation}
    \label{eq:4}
    \|fg\|\leq C\|A^rf\|\|A^{3/2-r}g\|.
  \end{equation}

\end{lemma}

Due to Lemma \ref{lemma:sobolev} we can estimate
\begin{equation}\label{eq:9}
  \|A^{5/4}\e^{-(t-s)A^2}P(Mu(s)\cdot \nabla)v(s)\|\leq
  C_1(t-s)^{-5/8}s^{-3/4}\|u\|_{I,3/8, 5/4}\,\|v\|_{I,3/8, 5/4}.
\end{equation} Here we used the spectral theorem to bound
$\|A^{5/4}\e^{-(t-s)A^2}\|_{\vB (L^2)}\leq C(t-s)^{-5/8}$ and the boundedness of
$P$ and $A^{1/4}\partial_jA^{-5/4}$.

Motivated by (\ref{eq:9}) let us write the integral equation
(\ref{eq:1}) as $X=Y+B(X,X)$ on the space  $\vB=\vB_{I,s_1,s_2}$ with
$s_1= 3/8$ and $s_2= 5/4$. Abbreviating  $|v|=\|v\|_{I,s_1,s_2}$ we
obtain that for all $u,v\in \vB$
\begin{equation}
  \label{eq:8}
 |B(u,v)|\leq \gamma |u|\,|v|;\;\gamma=C_1\sup_{t\in
   I}t^{3/8}\int_0^{t}(t-s)^{-5/8}s^{-3/4}\,\d s.
\end{equation} Notice that $\gamma$ does not
depend on $I$ since $C_1$ is the constant  coming from (\ref{eq:9})
and
\begin{equation}
  \label{eq:10}
 \gamma=C_1\int_0^{1}(1-s)^{-5/8}s^{-3/4}\,\d s.
\end{equation}

\subsection{Abstract scheme}\label{Abstract scheme} We shall study
(motivated  by (\ref{eq:8})) the equation
\begin{equation}
  \label{eq:5}
  X=Y+B(X,X),
\end{equation} where $B$ is a continuous bilinear operator $\vB\times
\vB\to \vB$ on a given Banach space $\vB$. Let $\gamma \geq 0$ denote a  corresponding  bounding
constant,
\begin{equation}
  \label{eq:7}
  |B(u,v)|\leq \gamma |u|\,|v|.
\end{equation} The elementary fixed point
theorem applies if there exists $R\geq 0$ such that
\begin{equation}
  \label{eq:6}
  |Y|\leq R\mand \kappa:=4\gamma R<1.
\end{equation} In fact letting $B_{\widetilde R}=\{X\in \vB|\,
|X|\leq {\widetilde R}\}$ for  ${\widetilde R}\geq0$  the conditions (\ref{eq:6}) assure that $B_{2R}$ is
mapped into itself by  the map
$X\to F(X):=Y+B(X,X)$ and that $\kappa$ is a
corresponding contraction constant.  (This version of the fixed point
theorem is implicitly used in \cite{Pl}, see \cite[Lemma 1]{Pl}.) In
particular under the condition (\ref{eq:6})  there exists a unique
solution to (\ref{eq:5}) in  $B_{2R}$. Letting  $X_0=Y$ and
$X_n=Y+B(X_{n-1},X_{n-1})$ for $n\in \N$  this solution can be represented
as
\begin{equation}
  \label{eq:2}
  X=\vB-\lim_{n\to \infty}X_n.
\end{equation}

\subsection{Local solvability in $\dot
H^r$  and  $H^r$,  $r\in [1/2, 3/2[$}  \label{Local solvability in dot r
and r for
large data} We return to the integral equation
(\ref{eq:1}) written as $X=Y+B(X,X)$ in the space  $\vB=\vB_{I,s_1,s_2}$ with
$s_1= 3/8$ and $s_2= 5/4$. We need to examine the first term
$Y(t)=\e^{-tA^2}u_0$ for some ``data'' $u_0\in (\vS')^3$. More precisely we need
to study the condition  $Y\in \vB$. Clearly due to (\ref{eq:8})
 the requirement (\ref{eq:6}) is met if $|Y|$ is sufficiently small.

Let us first examine the special case  $Y=\e^{-(\cdot)A^2}u_0\in
\vB_\infty$ where $\vB_\infty=\vB_{]0,\infty[,s_1,s_2}$. This requirement  is
equivalent to finiteness of the expression
\begin{equation}
  \label{eq:11}
  \sup_{q\in\Z}2^{q/2}\parbb {\int _{2^q\leq |\xi|<2^{q+1}}|\hat u_0|^2\,\d
  \xi}^{1/2},
\end{equation} cf. \cite[Lemma 8]{Pl}. Notice that  finiteness of
(\ref{eq:11}) and (\ref{eq:3}) are
equivalent to  finiteness of  these  expressions  for each of the three
components of $u_0$ and $v$, respectively.  For notational convenience
 we shall in the following  discussion
slightly
abuse notation by treating $u_0$ as a scalar-valued function rather than
an $\R^3$-valued function and similarly for the elements in  $\vB_\infty$. The (finite) expression  (\ref{eq:11})
is the norm of $u_0$ in the Besov space $\dot B^{1/2}_{2,\infty}$ which
indeed consists of all  $ u_0\in \vS'$ with $\hat u_0$ a
measurable function and (\ref{eq:11}) finite. In fact the
norms  (\ref{eq:3}) (with $\zeta=0$, $\theta=0$, $s_1= 3/8$ and $s_2= 5/4$) and
(\ref{eq:11}) are equivalent on the subspace  of $\vB_\infty$
consisting of functions
$t\to \e^{-tA^2}u_0$ where  $u_0\in \dot B^{1/2}_{2,\infty}$,
henceforth for
brevity denoted by $\e^{-(\cdot) A^2}\dot B^{1/2}_{2,\infty}$. Introducing $\dot
B^{{1/2},0}_{2,\infty}$ as  the set of $u_0\in \dot
B^{1/2}_{2,\infty}$ with
\begin{equation*}
2^{q/2}\parbb {\int _{2^q\leq |\xi|<2^{q+1}}|\hat u_0|^2\,\d
  \xi}^{1/2}\to 0\mfor q\to +\infty,
\end{equation*}
obviously
\begin{equation}\label{eq:13}
  \dot
H^{1/2}\subseteq \dot
B^{{1/2},0}_{2,\infty}\subseteq \dot B^{1/2}_{2,\infty}\subseteq
\cap_{1>\epsilon>0}\big (\dot
H^{1/2-\epsilon}+\dot
H^{1/2+\epsilon}\big ).
\end{equation}
Let $\vB^0_\infty$ be the subspace of
$\vB_\infty$ consisting of functions $v\in \vB_\infty$ obeying
\begin{equation}\label{eq:14}
  t^{3/8}\|A^{5/4}v(t)\|\to 0 \mfor t\to 0.
\end{equation} We have the following identification of subspaces in
$\vB_\infty$ (which is easily proven)
\begin{equation}\label{eq:12}
  \e^{-(\cdot) A^2}\dot B^{1/2,0}_{2,\infty}=\vB^0_\infty\cap\e^{-(\cdot) A^2}\dot B^{1/2}_{2,\infty}.
\end{equation}

Now returning to a general interval $I$ we introduce the subspace of
$\vB$, denoted by
$\vB^0$,  consisting of (vector-valued) functions  $v$ obeying (\ref{eq:14}).

\begin{prop}\label{prop:kato_small_data} Suppose $u_0\in(\dot
H^{1/2})^3$. Then for any  $T=|I|>0$ small enough (so that the conditions
(\ref{eq:6}) hold for some $R>0$)
 the
integral equation (\ref{eq:5})  has a unique solution  in the ball  $B_{2R}\subset\vB$
where $\vB=\vB_{I, s_1,s_2}$ with
$s_1= 3/8$, $s_2= 5/4$ and
 $Y(t)=\e^{-tA^2}u_0$. This solution $X\in \vB^0\cap BC(\bar I,(\dot
H^{1/2})^3)$ with $X(0)=u_0$. If in addition $u_0\in (L^2)^3$ then $X\in BC(\bar I,(H^{1/2})^3)$.
  \end{prop}
  \begin{proof}
    By combining (\ref{eq:8})
and (\ref{eq:12}) we conclude that the requirements (\ref{eq:6}) are  met
in the space $\vB$ provided that the three components of $u_0$
belong to $\dot B^{1/2,0}_{2,\infty}$
and that the parameter $T$ is
taken small enough (to ensure that $|Y|$ is small). Whence there
exists a unique solution $X\in B_{2R}$. We notice that  $X$ is also the unique solution
to the fixed point problem in the ball   $B_{2R}^0:=\vB^0\cap B_{2R}$ (if the components of $u_0$
belong to $
\dot B^{1/2,0}_{2,\infty}$ and $T>0$ is
 small).

Using  the first
inclusion of (\ref{eq:13}) we obtain in particular a unique small time
solution with ``data''  $u_0\in (\dot
H^{1/2})^3$ in $B_{2R}^0$. By an estimate very similar to
(\ref{eq:9}) we obtain the bound
\begin{equation}
  \label{eq:7b}
  \|B(u,v)(t)\|_{(\dot
H^{1/2})^3}\leq \eta |u|\,|v|\mfor t\in I.
\end{equation}
Using  (\ref{eq:7b}) we
see that in fact $B:\vB^0\times \vB^0\to BC(\bar I,(\dot
H^{1/2})^3)$ and that the functions in the range of this map  vanish at
$t=0$. Consequently,  the constructed fixed point $X\in \vB^0$ for $u_0\in (\dot
H^{1/2})^3$  belongs to the space $BC(\bar I,(\dot
H^{1/2})^3)$ and  the data $X(0)=u_0$ is attained continuously.

Moreover we have the bound
\begin{equation}
  \label{eq:7bb}
  \|B(u,v)(t)\|\leq C t^{1/4}|u|\,|v|\mfor t\in I.
\end{equation} So if in addition $u_0\in (L^2)^3$ then $X\in BC(\bar I,(L^2)^3)$ with   the data
$X(0)=u_0$  attained continuously.

\end{proof}

\begin{remark}\label{remark:kato-fu} In \cite{KF1,KF2}
  spaces $\vB$ and  $\vB^0$ similar to ours are  used for treating
  (\ref{eq:equamotiona}). The powers differ from ours:
  $t^{3/8}\to t^{1/8}$ and $A^{5/4}\to A^{3/4}$. These spaces are not
  suitable for the generalized problem (\ref{eq:equamotion}).

\end{remark}

If the data $u_0\in(\dot
H^{r})^3$ for $r\in ]1/2,3/2[$ there is a similar result (in fact, as
the reader will see,
proved very similarly).

We let
$\tilde r=\max(r, 5/4)$ and  introduce $\vB=\vB_{I, s_1,s_2}$  with
$s_1=\tilde r/2-r/4-1/8$ and $s_2=\tilde r$.

We use Lemma \ref{lemma:sobolev} to bound
\begin{equation}
  \label{eq:15}
  \|A^{2s_2-5/2}P(Mu(s)\cdot \nabla)v(s)\|\leq C \|A^{s_2}u(s)\|\, \|A^{s_2}v(s)\|,
\end{equation} and we obtain the following
analogue of (\ref{eq:9}) by splitting $A^{s_2}=A^{5/2-s_2}A^{2s_2-5/2}$:
\begin{equation}\label{eq:9b}
  \|A^{s_2}\e^{-(t-s)A^2}P(Mu(s)\cdot \nabla)v(s)\|\leq
  C_1(t-s)^{s_2/2-5/4}s^{-2s_1}|u|\,|v|.
\end{equation} Consequently we infer that
\begin{equation}
  \label{eq:8b}
 |B(u,v)|\leq \gamma |u|\,|v|;\;\gamma=C_1\sup_{t\in
   I}t^{s_1}\int_0^{t}(t-s)^{s_2/2-5/4}s^{-2s_1}\,\d s=C_2T^{r/4-1/8}.
\end{equation}

The space
$\vB^0$ is now defined to be the space  of functions $v\in \vB$
 obeying
\begin{equation}\label{eq:14b}
  t^{s_1}\|A^{s_2}v(t)\|\to 0 \mfor t\to 0,
\end{equation}
cf. (\ref{eq:14}).

As for $Y:=\e^{-(\cdot)A^2}u_0$ indeed $Y\in \vB^0$ with
\begin{equation}
  \label{eq:16}
  |Y|\leq C \sup_{t\in I}t^{r/4-1/8}\|u_0\|_{\dot
H^{r}}=C T^{r/4-1/8}\|u_0\|_{\dot
H^{r}}.
\end{equation} Here and henceforth we slightly abuse notation by
abbreviating $(\dot
H^{r})^3$ as $\dot
H^{r}$.  Similarly from this point on we shall for convenience frequently  abbreviate $(
H^{r})^3$ as $
H^{r}$ and $(L^2)^3$ as $L^2$, respectively. (Hopefully the
interpretation  will
be obvious in every concrete context.)

\begin{prop}\label{prop:kato small data2} Suppose $u_0\in\dot
H^{r}$ with $r\in ]1/2,3/2[$. Let $\tilde r=\max(r, 5/4)$, $s_1=\tilde
r/2-r/4-1/8$ and  $s_2=\tilde r$.  For any  $T=|I|>0$ small enough (so
that the conditions (\ref{eq:6}) hold for some $R>0$)
 the
integral equation (\ref{eq:5})  has a unique solution  in the ball $B_{2R}\subset\vB$
where $\vB=\vB_{I, s_1,s_2}$
 and
 $Y(t)=\e^{-tA^2}u_0$. This solution $X\in \vB^0\cap BC(\bar I,\dot
H^{r})$ with $X(0)=u_0$. If in addition $u_0\in L^2$ then $X\in
BC(\bar I,H^{r})$ (possibly we need at this point to take $T>0$
smaller if $r\in ]5/4,3/2[$).
  \end{prop}
  \begin{proof} Due to (\ref{eq:8b}) and the fact that $Y\in \vB^0$
    indeed there exists a unique solution $X\in B_{2R}$ for $T>0$
    small enough, and we  notice that  $X$ is also the unique solution
to the fixed point problem in the ball    $B_{2R}^0:=\vB^0\cap B_{2R}$.

   With the  convention alluded to above the analogue of (\ref{eq:7b}) reads
\begin{equation}
  \label{eq:7c}
  \|B(u,v)(t)\|_{\dot
    H^{r}}\leq \eta |u|\,|v|\mfor t\in I,
\end{equation} and we infer (as before) that   $B:\vB^0\times \vB^0\to BC( \bar I,\dot
H^{r})$ and that the functions in the range of this map vanish at
 $t=0$. (Note incidently that the constant $\eta$ of (\ref{eq:7c})
can be chosen independently of $T$ but not as a vanishing power of $T$
as in (\ref{eq:8b}) and (\ref{eq:16}).) Whence indeed $X\in \vB^0\cap BC(\bar I,\dot
H^{r})$ with $X(0)=u_0$.

Moreover for $r\in ]1/2,5/4]$ we have the bound
\begin{equation}
  \label{eq:7bbb}
  \|B(u,v)(t)\|\leq C t^{r/2}|u|\,|v|\mfor t\in I.
\end{equation} So if $r\in ]1/2,5/4]$  and in addition $u_0\in L^2$ then   $X\in BC(\bar I,L^2)$ with   the data $X(0)=u_0$
attained continuously. Whence $X\in BC(\bar I,H^r)$ for $u_0\in
H^r$. In fact this holds for any  $r\in ]1/2,3/2[$ (possibly by taking
a smaller interval $I$ if $r\in ]5/4,3/2[$). In the case
$r\in ]5/4,3/2[$  we can  obtain the result from the case $r=5/4$ by invoking
the embedding $H^r \subseteq H^{5/4}$ and using
the representation (\ref{eq:2}) of the solutions with data in $\dot
H^r $ and $\dot H^{5/4}$, respectively. We deduce that  for  data in $\dot
H^r \cap \dot H^{5/4}$  the two
constructed
solutions, say $X_1\in\vB_1$ and $X_2\in\vB_2$,   coincide on their common interval of
definition  $I_1\cap
I_2$.

\end{proof}

\subsection{Local solvability in $H^r$,  $r\in [5/4, \infty[$} \label{Local solvability in Hr,  rin 5/4, infty}
In this subsection we shall study local solutions in $H^r$ for $r\in [5/4,
\infty[$. The method of proof will be similar to that of Subsection \ref{Local solvability in dot r
and r for
large data}. In particular our constructions will be based on the
following modification of (\ref{eq:3}) (for simplicity we shall  use the
same notation).
\begin{equation}
  \label{eq:3mod}
  \|v\|_{\zeta, \theta,I,s_1,s_2}:=\sup_{t\in I}\e^{-\zeta(t)}t^{s_1}\|\e^{\theta(t)A}\inp{A}^{s_2}v(t)\|.
\end{equation} Again we consider here the case $\zeta=0$ and
$\theta=0$ only. Let $\tilde r\in [5/4,3/2[$ be arbitrarily given such
that $r\geq\tilde r$. Then the  parameters $s_1$ and $s_2$ in
(\ref{eq:3mod}) are chosen as follows:
\begin{equation}
  \label{eq:27}
  s_1=\tilde r/4-1/8\mand s_2=r,
\end{equation} and $\vB=\vB_{I,s_1,s_2}$ is the class of
functions $I\ni t\to v(t)\in  H^{s_2}$ for which the expression
$t^{s_1}\inp{A}^{s_2}v(t)$ defines an element in $BC(I,L^2)$;
$I=]0,T]$.

To bound $B:\vB\times \vB\to \vB$ we let
\begin{equation}
  \label{eq:28}
 \bar r=5/2-\tilde r,
\end{equation}
and split
$\inp{A}^{s_2}=\inp{A}^{\bar r}\inp{A}^{r-\bar r}$. Using again Lemma
\ref{lemma:sobolev} we then obtain
\begin{equation}\label{eq:9bmod}
  \|\inp{A}^{s_2}\e^{-(t-s)A^2}P(Mu(s)\cdot \nabla)v(s)\|\leq
  C_1\inp{T}^{\bar r/2}(t-s)^{-\bar r/2}s^{-2s_1}|u|\,|v|.
\end{equation} Consequently we infer that
\begin{equation}
  \label{eq:8bmod}
 |B(u,v)|\leq \gamma |u|\,|v|;\;\gamma=C_1\inp{T}^{\bar r/2}\sup_{t\in
   I}t^{s_1}\int_0^{t}(t-s)^{-\bar r/2}s^{-2s_1}\,\d s=C_2\inp{T}^{\bar r/2}T^{s_1}.
\end{equation}

The space
$\vB^0$ is the subclass of $v\in \vB$ obeying
\begin{equation}\label{eq:14bmod}
  t^{s_1}\|\inp{A}^{s_2}v(t)\|\to 0 \mfor t\to 0.
\end{equation}

Now suppose $u_0\in H^r$. Then clearly $Y:=\e^{-(\cdot)A^2}u_0\in \vB^0$ with
\begin{equation}
  \label{eq:16mod}
  |Y|\leq C \sup_{t\in I}t^{s_1}\|u_0\|_{
H^{r}}=C T^{s_1}\|u_0\|_{
H^{r}}.
\end{equation}

Next letting (as before) $B_{2R}:=\{X\in \vB|\,
|X|\leq 2 |Y|\}$ and $B_{2R}^0:=\vB^0\cap B_{2R}$ we conclude from
(\ref{eq:8bmod}) and the fact that $Y\in \vB^0$ that indeed the contraction condition (\ref{eq:6})
for the map $X\to Y+B(X,X)$ restricted to either  $B_{2R}$ or to
$B_{2R}^0$  is
valid for any $T>0$ small enough. Consequently the common fixed point $X\in \vB^0$. Moreover
\begin{equation}
  \label{eq:7cmod}
  \|B(u,v)(t)\|_{
    H^{r}}\leq \eta\inp{T}^{\bar r/2} |u|\,|v|\mfor t\in I,
\end{equation} and we infer (as before) that   $B:\vB^0\times \vB^0\to BC( \bar I,
H^{r})$ and that the functions in the range of this map vanish at
 $t=0$. We conclude:

\begin{prop}\label{prop:kato small data2mod} Suppose $u_0\in
H^{r}$ with $r\in [5/4,\infty[$. Let $\tilde r\in [5/4,3/2[$ be arbitrarily given such
that $r\geq\tilde r$. Let $s_1$ and $s_2$ be given as in
(\ref{eq:27}).  For any  $T=|I|>0$ small enough (so that the
conditions (\ref{eq:6})
hold for some $R>0$)
 the
integral equation (\ref{eq:5})  has a unique solution  in the ball  $
B_{2R}\subset \vB$
where $\vB=\vB_{I, s_1,s_2}$
 and
 $Y(t)=\e^{-tA^2}u_0$. This solution $X\in \vB^0\cap BC(\bar I,
H^{r})$ with $X(0)=u_0$.
  \end{prop}

  \begin{remarks}
    \label{remarks:uniq}
    \begin{enumerate}[\normalfont 1)]
    \item \label{item:1} Clearly we may take $\tilde r=r$ in
      Proposition \ref{prop:kato small data2mod} if $r\in
      [5/4,3/2[$. In that case the (small) $T$-dependence of the bounds
      (\ref{eq:16}) and (\ref{eq:8b}) coincides with that of
      (\ref{eq:16mod}) and (\ref{eq:8bmod}), respectively.
\item \label{item:2}  Although there are different spaces involved in
Propositions \ref{prop:kato_small_data}--\ref{prop:kato small data2mod}
the constructed solutions coincide on any common interval of
definition, cf.  the last
argument of the proof of  Proposition
\ref{prop:kato small data2}.
    \end{enumerate}
\end{remarks}

  \subsection{Global solvability in $\dot H^{1/2}$ and $H^{1/2}$ for
    small data} \label{Global solvability in dot 1/2 for small data}
  If the data in $\dot H^{1/2}$  is small we can improve on the
  conclusion of Proposition \ref{prop:kato_small_data}  to obtain  a
  global solution. The proof is similar.

\begin{prop}\label{prop:kato big data} Suppose $u_0\in\dot
H^{1/2}$.
  Then the
integral equation (\ref{eq:5})  has a unique solution  in the ball
$B_{2R}\subset \vB$
where $\vB=\vB_{I, s_1,s_2}$ with $I=]0,\infty[$,
$s_1= 3/8$, $s_2= 5/4$ and
 $Y(t)=\e^{-tA^2}u_0$ provided  that $\|u_0\|_{\dot H^{1/2}}$ is
 sufficiently small (so that for some $R>0$  and with
$\gamma$ given by (\ref{eq:10}) the conditions (\ref{eq:6}) hold). This solution $X\in \vB^0\cap BC([0,\infty[,\dot
H^{1/2})$ with $X(0)=u_0$. If in addition $u_0\in L^2$ then $X\in C([0,\infty[,H^{1/2})$.
  \end{prop}

\section{Analyticity bounds for small times} \label{Analyticity bounds
  for small times}
In this section we shall study analyticity properties of the
short-time solutions of Propositions \ref{prop:kato_small_data}--\ref{prop:kato small data2mod}. This will be done by using more general
spaces with norms given by (\ref{eq:3}) or (\ref{eq:3mod})  and modifying the
proofs of Section \ref{Integral equation}.
\subsection{Local analyticity bounds in $\dot H^r$  and  $H^r$ ,  $r\in [1/2, 3/2[$} \label{Analyticity bounds in
    dot r and  r
    for small times and for large data}
In this subsection we specify the
functions $\zeta$ and $\theta$  in the norm
(\ref{eq:3}) in terms of a parameter $\lambda\geq 0$ as
\begin{equation}
  \label{eq:20}
  \zeta=\lambda^2/4+(\tilde r-r)\ln
\inp{\lambda}\mand \theta(t)=\lambda\sqrt t;
\end{equation}
  here $\tilde r$ is given as in Proposition
\ref{prop:kato small data2} (also for $r=1/2$). Let $s_1$ and $s_2$ be
given as  in Proposition
\ref{prop:kato small data2}, and let again $I=]0,T]$.

For applications in Subsection \ref{Improved analyticity
    bounds in dot r and  r, r>
  1/2,
    for small times and for large data} we shall  be concerned below
  with bounding various quantities independently of the parameter
  $\lambda\geq0$ (rather than just proving Theorem
  \ref{prop:ana_small_times1} stated below). We shall use the following elementary bound (which follows from the
spectral theorem).
\begin{lemma}\label{Lemma:spect_bnd} For any $\alpha\geq 0$ there
  exists a constant $C\geq 0$ such that for all
  $f\in L^2$
  \begin{equation}
    \label{eq:17}
    \sup_{\lambda, t\geq 0}\ \inp{\lambda}^{-\alpha}\e^{-\lambda^2/4}\|\parb{\sqrt t A}^\alpha\e^{\lambda\sqrt
      tA}\e^{-tA^2}f\|\leq C\|f\|.
  \end{equation}
\end{lemma}

Due to (\ref{eq:17}) we can estimate the following norm of
$Y=\e^{-(\cdot)A^2}u_0$, $u_0\in \dot H^r$.
\begin{equation}
  \label{eq:18}
  |Y|=\|Y\|_{\zeta, \theta,I,s_1,s_2}\leq C_1 T^{r/4-1/8}\|A^ru_0\|,
\end{equation} where  the constant $C_1$ is  independent of
$\lambda\geq0$.

Let $\vB^0=\vB^0_{\zeta, \theta,I,s_1,s_2}$ be the space  of functions $v\in \vB:=\vB_{\zeta, \theta,I,s_1,s_2}$
 obeying
\begin{equation}\label{eq:14bb}
  t^{s_1}\|\e^{\lambda\sqrt
      tA}A^{s_2}v(t)\|\to 0 \mfor t\to 0,
\end{equation} Obviously it follows from (\ref{eq:18}) that
$Y\in\vB^0$ if $r>1/2$. However this is also true for $r=1/2$ which
 follows from the same bound and a simple approximation
 argument (using for instance that $H^{5/4}$ is dense in $\dot H^{1/2}$). Whence indeed
 \begin{equation}
   \label{eq:19}
   Y=\e^{-(\cdot)A^2}u_0\in\vB^0\text{ for }u_0\in \dot H^r,\ r \in [1/2, 3/2[.
 \end{equation}

We have the following generalization of  Propositions \ref{prop:kato_small_data} and
\ref{prop:kato small data2} (abbreviating as before the norm on  $\vB$
as $|\cdot|$). Notice that the solutions of Theorem
\ref{prop:ana_small_times1} coincide with
those of Propositions \ref{prop:kato_small_data} and
\ref{prop:kato small data2} for $T>0$  small enough, cf. Remark
\ref{remarks:uniq} \ref{item:2}.
\begin{thm}
  \label{prop:ana_small_times1} Let   $\lambda\geq0$ and $u_0\in\dot
H^{r}$ with $r\in [1/2,3/2[$ be given. Let $\tilde r=\max(r, 5/4)$, $s_1=\tilde
r/2-r/4-1/8$ and  $s_2=\tilde r$.  For any  $T=|I|>0$ small enough (so
that the conditions (\ref{eq:6}) hold for some $R>0$)
 the
integral equation (\ref{eq:5})  has a unique solution  in the ball
$B_{2R}\subset \vB$
where $\vB=\vB_{\zeta, \theta,I, s_1,s_2}$
 and
 $Y(t)=\e^{-tA^2}u_0$. This solution $X\in \vB^0\subseteq \vB$, and it obeys that $\e^{\lambda\sqrt{
      (\cdot)}A}X\in BC(\bar I,\dot
H^{r})$ with $X(0)=u_0$. If in addition $u_0\in L^2$ then $\e^{\lambda\sqrt{
      (\cdot)}A}X\in
BC(\bar I,H^{r})$ (possibly we need at this point to take $T>0$
smaller if $r\in ]5/4,3/2[$).
\end{thm}
\begin{proof} First we show that $B:\vB\times \vB\to \vB$. Using
  the triangle inequality in Fourier space we obtain the following
  analogue of (\ref{eq:15}):
\begin{equation}
  \label{eq:15ba}
  \|\e^{\lambda\sqrt
      sA}A^{2s_2-5/2}P(Mu(s)\cdot \nabla)v(s)\|\leq C \|\e^{\lambda{\sqrt
      s}A}A^{s_2}u(s)\|\, \|\e^{\lambda\sqrt
      sA}A^{s_2}v(s)\|,
\end{equation}  and consequently using that $\sqrt{t} \leq
\sqrt{t-s}+\sqrt{s}$ and Lemma \ref{Lemma:spect_bnd} we obtain, cf. (\ref{eq:9b}),
\begin{equation}\label{eq:9bba}
  \|\e^{\lambda\sqrt
      tA}A^{s_2}\e^{-(t-s)A^2}P(Mu(s)\cdot \nabla)v(s)\|\leq
  C_2\inp{\lambda}^{5/2-s_2}\e^{\lambda^2/4}\e^{2\zeta}(t-s)^{s_2/2-5/4}s^{-2s_1}|u|\,|v|.
\end{equation} We conclude, cf. (\ref{eq:8b}), that indeed $B(u,v)\in
\vB$ with
\begin{equation}
  \label{eq:8bba}
|B(u,v)|\leq \gamma |u|\,|v|;\;\gamma=C_3\inp{\lambda}^{5/2-r}\e^{\lambda^2/2}T^{r/4-1/8}.
\end{equation}

  Obviously the same arguments show that $B:\vB^0\times \vB^0\to
  \vB^0$. In conjunction with (\ref{eq:19}) we conclude that the integral
  equation (\ref{eq:5})  has a unique solution $X$  in the ball
  $B_{2R}$ provided  that first $R>0$ and then $T>0$ are  taken small enough, and that this $X\in \vB^0$.

For the remaining statements of Theorem
\ref{prop:ana_small_times1} we similarly mimic the proof of
Proposition
\ref{prop:kato small data2}.
 For completeness of presentation  let us
state the analogues of (\ref{eq:7c}) and (\ref{eq:7bbb})
\begin{equation}
  \label{eq:7cq}
  \|\e^{\lambda \sqrt tA}B(u,v)(t)\|_{\dot
    H^{r}}\leq C \inp{\lambda}^{5/2-r}\e^{ 3\lambda^2/4}|u|\,|v|\mfor t\in I,
\end{equation} and for $r\in [1/2,5/4]$
\begin{equation}
  \label{eq:7bbbq}
  \|\e^{\lambda \sqrt tA}B(u,v)(t)\|\leq \widetilde C \inp{\lambda}^{5/2-2r}\e^{ 3\lambda^2/4}t^{r/2}|u|\,|v|\mfor t\in I.
\end{equation}
\end{proof}

\subsection{Local analyticity bounds in $H^r$, $r\in [5/4, \infty[$} \label{Analyticity bounds in
    r
    for small times and for large data}
In this subsection we assume $r\in [5/4, \infty[$ and specify the
functions $\zeta$ and $\theta$  in the norm
(\ref{eq:3mod}) in terms of a parameter $\lambda\geq 0$ as
\begin{equation}
  \label{eq:20 mod}
  \zeta=\lambda^2/4\mand \theta(t)=\lambda\sqrt t.
\end{equation} Let $s_1$ and $s_2$ be
given as  in Proposition
\ref{prop:kato small data2mod}, and let again $I=]0,T]$.

Due to Lemma \ref{Lemma:spect_bnd}   we can estimate the following norm of
$Y=\e^{-(\cdot)A^2}u_0$, $u_0\in  H^r$.
\begin{equation}
  \label{eq:18mod_anal}
  |Y|=\|Y\|_{\zeta, \theta,I,s_1,s_2}\leq C_1 T^{s_1}\|\inp{A}^ru_0\|,
\end{equation} where  the constant $C_1$ is  independent of
$\lambda\geq0$.

Let $\vB^0=\vB^0_{\zeta, \theta,I,s_1,s_2}$ be the space  of functions $v\in \vB:=\vB_{\zeta, \theta,I,s_1,s_2}$
 obeying
\begin{equation}\label{eq:14bbmod}
  t^{s_1}\|\e^{\lambda\sqrt
      tA}\inp{A}^{s_2}v(t)\|\to 0 \mfor t\to 0,
\end{equation} It follows from (\ref{eq:18mod_anal}) that
$Y=\e^{-(\cdot)A^2}u_0\in\vB^0$.

We have the following generalization of  Proposition \ref{prop:kato small data2mod}:
\begin{thm}
  \label{prop:ana_small_times1mod} Let   $\lambda\geq0$ and $u_0\in
H^{r}$ with $r\in [5/4, \infty[$ be given. Let $\tilde r\in [5/4,3/2[$ be arbitrarily given such
that $r\geq\tilde r$. Let $s_1$ and $s_2$ be given as in
(\ref{eq:27}).  For any  $T=|I|>0$ small enough (so that the
conditions (\ref{eq:6})
hold for some  $R>0$)
 the
integral equation (\ref{eq:5})  has a unique solution  in the ball
$B_{2R}\subset \vB$
where $\vB=\vB_{\zeta, \theta,I, s_1,s_2}$
 and
 $Y(t)=\e^{-tA^2}u_0$. This solution $X\in \vB^0\subseteq \vB$, and it
 obeys  that $\e^{\lambda\sqrt{
      (\cdot)}A}X\in BC(\bar I,
H^{r})$ with $X(0)=u_0$.
\end{thm}
\begin{proof}
  Mimicking  the proof of Proposition \ref{prop:kato small data2mod}
  we obtain
\begin{equation}\label{eq:9bbamod2}
  \|\e^{\lambda\sqrt
      tA}\inp{A}^{s_2}\e^{-(t-s)A^2}P(Mu(s)\cdot \nabla)v(s)\|\leq
  C_2\inp{\lambda}^{\bar r}\e^{\lambda^2/4}\e^{2\zeta}\inp{t}^{\bar r/2}(t-s)^{-\bar r/2}s^{-2s_1}|u|\,|v|,
\end{equation} where $\bar r$ is given by (\ref{eq:28}).
Consequently we infer that
\begin{align}
  \label{eq:8bmod_anal}
 |B(u,v)|&\leq \gamma |u|\,|v|;\\ \gamma&=C_2\inp{\lambda}^{\bar r}\e^{\lambda^2/2}\inp{T}^{\bar r/2}\sup_{t\in
   I}t^{s_1}\int_0^{t}(t-s)^{-\bar r/2}s^{-2s_1}\,\d s=C_3\inp{\lambda}^{\bar r}\e^{\lambda^2/2}\inp{T}^{\bar r/2}T^{s_1}. \nonumber
\end{align}

We conclude from (\ref{eq:18mod_anal})  and
(\ref{eq:8bmod_anal}) that indeed the contraction condition (\ref{eq:6})
for the map $X\to Y+B(X,X)$ on $B_{2R}$ (or on
$B_{2R}^0:=\vB^0\cap B_{2R}$)  is
valid for any $T>0$ small enough. Moreover
\begin{equation}
  \label{eq:7cmod_anal}
  \|\e^{\lambda \sqrt tA}B(u,v)(t)\|_{
    H^{r}}\leq C\inp{\lambda}^{\bar r}\e^{3\lambda^2/4} \inp{T}^{\bar
    r/2} |u|\,|v|\mfor t\in I,
\end{equation} and we infer (as before) that   $\e^{\lambda\sqrt{
      (\cdot)}A}B:\vB^0\times \vB^0\to BC(\bar I,
H^{r})$ and that the functions in the range of this map  vanish at
$t=0$.
\end{proof}

\subsection{Improved local bounds of analyticity radii  for $r>
  1/2$} \label{Improved analyticity
    bounds in dot r and  r, r>
  1/2,
    for small times and for large data}
In this subsection we shall modify the constructions of Subsections \ref{Analyticity bounds in
    dot r and  r
    for small times and for large data} and \ref{Analyticity bounds in
    r
    for small times and for large data} in that the functions in
  (\ref{eq:20}) and (\ref{eq:20 mod}) now will
  be taken with an additional time-dependence. Explicitly we define $\zeta$ and
  $\theta$ by (\ref{eq:20}) (for the setting of Subsection \ref{Analyticity bounds in
    dot r and  r
    for small times and for large data})  and (\ref{eq:20 mod}) (for the setting of Subsection \ref{Analyticity bounds in
    r
    for small times and for large data}) but now in terms of  $\lambda$  taken to have
  the following explicit time-dependence
  \begin{equation}
    \label{eq:21}
    \lambda=\lambda_0\sqrt{t/T};
  \end{equation}
   here $\lambda_0\geq0$ is an auxiliary
  parameter (which in the end will play the role of the previous parameter
   $\lambda$) and $T>0$ is the right end point of the interval $I$ (as
  in Subsection \ref{Analyticity bounds in
    dot r and  r
    for small times and for large data}). The bounds (\ref{eq:18}),
  (\ref{eq:15ba}), (\ref{eq:18mod_anal}) and (\ref{eq:9bbamod2})  remain
  true where $\lambda =
  \lambda(s)$ in (\ref{eq:15ba})  and similar interpretations are needed in (\ref{eq:9bbamod2}).

As for (\ref{eq:8bba}) the bounding
  constant  has
  the form, cf. (\ref{eq:9bba}),
  \begin{equation}
    \label{eq:22}
    C_2\sup_{t\in
   I}\e^{-\zeta(t)}t^{s_1}\int_0^{t}\inp{\lambda(t-s)}^{5/2-s_2}\e^{\lambda^2(t-s)/4}\e^{2\zeta(s)}(t-s)^{s_2/2-5/4}s^{-2s_1}\
 \d s.  \end{equation} Since $\inp{\lambda (t_1)}^\alpha\leq
\inp{\lambda (t_2)}^\alpha$ if  $0\leq t_1\leq t_2\leq T$ and $\alpha\geq0$, and
\begin{equation}
  \label{eq:23}
  \e^{-\lambda^2(t)/4}\e^{\lambda^2(t-s)/4}\e^{\lambda^2(s)/2}=\e^{\lambda^2(s)/4}\leq
  \e^{\lambda_0^2/4},
\end{equation} we obtain
\begin{subequations}
 \begin{equation}
  \label{eq:24}
  |B(u,v)|\leq \gamma |u|\,|v|\text{ with }\gamma=C_3\inp{\lambda_0}^{5/2-r}\e^{\lambda_0^2/4}T^{r/4-1/8}.
\end{equation} (Notice that the cancellation (\ref{eq:23}) accounts for
the ``improvement'' $\e^{\lambda^2/2} \to \e^{\lambda_0^2/4}$ compared to
(\ref{eq:8bba}).)

As for (\ref{eq:7cq}) and (\ref{eq:7bbbq}) we obtain similar ``improvements''
\begin{equation}
  \label{eq:7cql}
  \|\e^{\theta(t)A}B(u,v)(t)\|_{\dot
    H^{r}}\leq C \inp{\lambda_0}^{5/2-r}\e^{ \lambda_0^2/2}|u|\,|v|\mfor t\in I,
\end{equation} and for $r\in [1/2,5/4]$
\begin{equation}
  \label{eq:7bbbql}
  \|\e^{\theta(t)A}B(u,v)(t)\|\leq \widetilde C \inp{\lambda_0}^{5/2-2r}\e^{ \lambda^2_0/2}t^{r/2}|u|\,|v|\mfor t\in I.
\end{equation}
 \end{subequations}

Arguing similarly for the setting of Subsection \ref{Analyticity bounds in
    r
    for small times and for large data} we obtain in this case
  \begin{subequations}
   \begin{equation}
  \label{eq:24mod}
  |B(u,v)|\leq \gamma |u|\,|v|\text{ with
  }\gamma=C_3\inp{\lambda_0}^{\bar r}\e^{\lambda_0^2/4}\inp{T}^{\bar r/2}T^{s_1},
\end{equation} and
\begin{equation}
  \label{eq:7cqlmod}
  \|\e^{\theta(t)A}B(u,v)(t)\|_{
    H^{r}}\leq C \inp{\lambda_0}^{\bar r}\e^{ \lambda_0^2/2}\inp{T}^{\bar r/2}|u|\,|v|\mfor t\in I.
\end{equation}
 \end{subequations}

Let us now investigate  the conditions (\ref{eq:6}) with  $R=|Y|$: Due to
(\ref{eq:18}) and (\ref{eq:24}) it suffices to have
\begin{equation}
  \label{eq:25}
  C_4\inp{\lambda_0}^{5/2-r}\e^{\lambda_0^2/4}T^{r/2-1/4}<1;\
  C_4=4C_1
  C_3\|A^ru_0\|.
\end{equation} Notice that the constants $C_1$ and $ C_3$ from
(\ref{eq:18}) and (\ref{eq:24}), respectively, are independent of
$\lambda_0$, $T$ and $u_0$. Therefore also $C_4$ is independent of
$\lambda_0$ and  $T$.

Now, assuming $r\in]1/2,3/2[$, we fix
$\epsilon \in ]0,2r-1]$.  Taking  then
\begin{equation}
  \label{eq:26}
  \lambda_0=\sqrt{2r-1-\epsilon}\ \sqrt{|\ln T|}
\end{equation} indeed (\ref{eq:25}) is valid provided $T>0$ is  small
enough, viz.
$T\leq T_0=T\big (\epsilon, r,\tilde r, \|A^ru_0\|\big )$.

Similarly in the setting of Subsection \ref{Analyticity bounds in
    r
    for small times and for large data} (due to (\ref{eq:18mod_anal}) and
  (\ref{eq:24mod}))  the conditions (\ref{eq:6}) with  $R=|Y|$ are
  valid if
\begin{equation}
  \label{eq:25_mod}
  C_4\inp{\lambda_0}^{\bar r}\e^{\lambda_0^2/4}\inp{T}^{\bar r/2}T^{2s_1}<1;\
  C_4=4C_1
  C_3\|\inp{A}^ru_0\|.
\end{equation} Fix
$\epsilon \in ]0,2\tilde r-1]$.  Taking  then
\begin{equation}
  \label{eq:26mod}
  \lambda_0=\sqrt{2\tilde r-1-\epsilon}\ \sqrt{|\ln T|}
\end{equation} the bound (\ref{eq:25_mod}) is valid provided
$T\leq T_0=T\big (\epsilon, r,\tilde r,\|\inp{A}^ru_0\|\big )$.

We have (almost) proved:
\begin{thm}
  \label{thm:impr_ana_small_times1}
  \begin{enumerate}[\normalfont i)]
  \item \label{item:3} Suppose    $u_0\in\dot
H^{r}$ for some  $r\in ]1/2,3/2[$. Put $\tilde r=\max(r, 5/4)$, $s_1=\tilde
r/2-r/4-1/8$ and  $s_2=\tilde r$, and let $\epsilon\in ]0,2r-1]$.  There exists $T_0=T\big (\epsilon, r,\tilde r, \|A^ru_0\|\big )>0$ such
that for any  $T\in]0,T_0]$
the integral equation (\ref{eq:5})  has a unique solution  in the ball
$B_{2|Y|}\subseteq \vB$
where $\vB=\vB_{\zeta, \theta,I, s_1,s_2}$ has norm (\ref{eq:3}) with  $\zeta$ and
  $\theta$ given by (\ref{eq:20}), (\ref{eq:21}) and (\ref{eq:26}),
  $I=]0,T]$
 and
 $Y(t)=\e^{-tA^2}u_0$. This solution $X\in \vB^0$, and it  obeys that $\e^{\theta A}X\in BC(\bar I,\dot
H^{r})$ with $X(0)=u_0$. If in addition $u_0\in L^2$ and $r\in
]1/2,5/4]$ then $\e^{\theta A}X\in
BC(\bar I,H^{r})$.

Moreover there are bounds
\begin{subequations}
  \begin{equation}
  \label{eq:7cqla}
  \|\e^{\theta(t)A}X(t)\|_{\dot
    H^{r}}\leq C T^{-(r/2-1/4-\epsilon/4)}\mfor t\in I,
\end{equation} and assuming
 in addition $u_0\in L^2$ and $r\in ]1/2,5/4]$
\begin{equation}
  \label{eq:7bbbqla}
  \|\e^{\theta(t)A}X(t)\|\leq  \widetilde CT^{-(r/2-1/4-\epsilon/4)}\mfor t\in I.
\end{equation}
\end{subequations}
The dependence of the constants $C$ and $\widetilde C$
on $u_0$ is through the quantity $\|A^ru_0\|$ and through the
quantities $\|A^ru_0\|$ and $\|u_0\|$, respectively.
\item\label{item:4}
Suppose   $u_0\in
H^{r}$ for some  $r\in [5/4,\infty[$. Let $\tilde r\in [5/4,3/2[$ be given such
that $r\geq\tilde r$ and let $\epsilon\in ]0,2\tilde r-1]$. Put
$s_1=\tilde r/4-1/8$ and $s_2=r$.  There exists $T_0=T\big (\epsilon, r,\tilde r, \|\inp{A}^ru_0\|\big )>0$ such
that for any  $T\in]0,T_0]$
the integral equation (\ref{eq:5})  has a unique solution  in the ball
$B_{2|Y|}\subseteq \vB$
where $\vB=\vB_{\zeta, \theta,I, s_1,s_2}$ has norm
(\ref{eq:3mod}) with  $\zeta$ and
  $\theta$ given by (\ref{eq:20 mod}), (\ref{eq:21}) and (\ref{eq:26mod}),
  $I=]0,T]$
 and
 $Y(t)=\e^{-tA^2}u_0$. This solution $X\in \vB^0$, and it obeys that $\e^{\theta A}X\in BC(\bar I,
H^{r})$ with $X(0)=u_0$.

Moreover
\begin{equation}
  \label{eq:7cqlamod}
  \|\e^{\theta(t)A}X(t)\|_{
    H^{r}}\leq C T^{-(\tilde r/2-1/4-\epsilon/4)}\mfor t\in I.
\end{equation} The dependence of the constant $C$
on $u_0$ is through the quantity $\|\inp{A}^ru_0\|$.
  \end{enumerate}
\end{thm}
\begin{proof}
  We use the arguments preceeding the
  theorem to get (unique) solutions.

As for \ref{item:3} the bounds (\ref{eq:7cqla}) and (\ref{eq:7bbbqla}) follow
  from (\ref{eq:7cql}) and (\ref{eq:7bbbql}), respectively, and
  (\ref{eq:18}) and Lemma
  \ref{Lemma:spect_bnd} (taking $\alpha=0$ there). Notice that the
  contributions from the non-linear term $B(X,X)$ have better bounds.

As for \ref{item:4} the bound (\ref{eq:7cqlamod}) follows
  from (\ref{eq:18mod_anal}), (\ref{eq:7cqlmod}) and Lemma
  \ref{Lemma:spect_bnd} (taking again $\alpha=0$ there). Again the
  contribution from the non-linear term has a better bound.
\end{proof}

By choosing $t=T$  in (\ref{eq:7cqla})--(\ref{eq:7cqlamod})  we obtain:
\begin{cor}
  \label{cor:small} Under the conditions of Theorem
\ref{thm:impr_ana_small_times1} \ref{item:3} the  solution $X$
obeys
\begin{subequations}
 \begin{equation}
  \label{eq:7cqlay}
  \|\e^{\sqrt{2r-1-\epsilon}\ \sqrt{|\ln t|}\sqrt tA}X(t)\|_{\dot
    H^{r}}\leq Ct^{-(r/2-1/4-\epsilon/4)}\mforall t \in ]0,T_0],
\end{equation} and assuming
 in addition $u_0\in L^2$ and $r\in ]1/2,5/4]$
\begin{equation}
  \label{eq:7bbbqlay}
  \|\e^{\sqrt{2r-1-\epsilon}\ \sqrt{|\ln t|}\sqrt tA}X(t)\|\leq \widetilde Ct^{-(r/2-1/4-\epsilon/4)}\mforall  t\in ]0,T_0].
\end{equation}
\end{subequations}

Under the conditions of Theorem
\ref{thm:impr_ana_small_times1} \ref{item:4} the  solution $X$
obeys
\begin{equation}
  \label{eq:7cqlaymod}
  \|\e^{\sqrt{2\tilde r-1-\epsilon}\ \sqrt{|\ln t|}\sqrt tA}X(t)\|_{
    H^{r}}\leq Ct^{-(\tilde r/2-1/4-\epsilon/4)}\mforall t \in ]0,T_0].
\end{equation}
\end{cor}

Clearly the dependence of the constant $\widetilde C$ in
(\ref{eq:7bbbqlay}) of $u_0\in H^r$ can be taken through its norm
$\|\inp{A}^ru_0\|$. Moreover, if $u_0\in H^r$ for some $r\in [5/4,3/2[$ we can choose $\tilde r =r$
in (\ref{eq:7cqlaymod}). Whence in particular we obtain from Corollary
\ref{cor:small}:
\begin{cor}
  \label{cor:simple} Suppose $u_0\in H^r$ for some $r\in
  ]1/2,3/2[$. Let $X$ be the solution to (\ref{eq:1})  with
  initial data  $u_0$ as given in  Proposition
  \ref{prop:kato small data2}, and let $\epsilon\in
  ]0,2r-1]$. Then there exist constants $C_0=C\big (\epsilon, r,
  \|\inp{A}^ru_0\|\big )>0$  and $T_0=T\big (\epsilon, r,
  \|\inp{A}^ru_0\|\big )>0$ such
that
  \begin{equation}
    \label{eq:29}
    \|\e^{\sqrt{2r-1-\epsilon}\ \sqrt{t|\ln t|}A}X(t)\|_{
    H^{r}}\leq C_0t^{-(r/2-1/4-\epsilon/4)}\mforall t\in]0,T_0].
  \end{equation} In particular, using notation from Subsection
  \ref{Discussion of uniform
 real  analyticity}, for this solution to (\ref{eq:1})
  \begin{equation}
    \label{eq:46}
 \liminf_{t\to 0}\tfrac{\rad(X(t))}{\sqrt{t|\ln t|}}\geq \sqrt{2r-1}.
  \end{equation}
\end{cor}

\subsection{Discussion} \label{Example1}
 We are only allowed to  put $\tilde r =r$ in (\ref{eq:7cqlaymod}) if  $r\in [5/4,3/2[$ due to the restriction $\tilde r<3/2$  of
 Theorem \ref{thm:impr_ana_small_times1} \ref{item:4}.

One may
 conjecture that  also in the case $r> 3/2$ the quantities
 \begin{equation}\label{eq:30}
   \|\e^{\sqrt{2 r-1-\epsilon}\ \sqrt{t|\ln t|}A}X(t)\|_{
    H^{r}};\ \epsilon\in ]0,2r-1],
 \end{equation} are all finite for  $t$ sufficiently small (given that $u_0\in H^r$). However  the proof  for
 $ r\in [5/4,3/2[$ does not provide any indication.

On the other hand if $u_0$ is ``much smoother'', in fact analytic, there is indeed
an improvement. Suppose there exists $\theta_0 >0$ such that
one of the following conditions holds
\begin{enumerate}[A)]
\item\label{item:5}  In addition to the assumptions of Theorem
\ref{thm:impr_ana_small_times1} \ref{item:3} $\e^{\theta_0 A}u_0\in
\dot H^r$ (and possibly $\e^{\theta_0 A}u_0\in
L^2$).
\item\label{item:6} In addition to the assumptions of Theorem
\ref{thm:impr_ana_small_times1} \ref{item:4} $\e^{\theta_0 A}u_0\in
H^r$.
\end{enumerate}
 Then we have the following version of Corollary
\ref{cor:small}: With \ref{item:5}  we can in (\ref{eq:7cqlay}) replace $X(t)$
by $\e^{\theta_0 A}X(t)$ (and similarly in (\ref{eq:7bbbqlay})). With
\ref{item:6}   we can in (\ref{eq:7cqlaymod}) replace $X(t)$
 by $\e^{\theta_0 A}X(t)$.

The proof of these statements goes  along the same line as  the
previous ones. Notice that we only have to check the previous
 proofs with $\theta\to \theta+\theta_0$ in various bounds. In
 particular this replacement is introduced in the construction of
 Banach spaces.

Another point to be discussed is the limit $r\to 1/2$. Clearly the
expression  $\sqrt{2 r-1}$, related to   (\ref{eq:30}),  vanishes in this
  limit. One may ask if this kind of behaviour is expected. Clearly
  this question is related to whether Theorem
  \ref{prop:ana_small_times1} should be considered as being
  ``optimal'' for $r=1/2$. There is a partial affirmative answer to
  the latter given by an example: We shall construct a specific
  (classical) solution to (\ref{eq:equamotion}) for
 specific $M$ and $P$ such that
\begin{equation}
  \label{eq:31}
  u\in BC([0,\infty[,H^r) \mand u_0\in H^r\cap  \dot B^{1/2}_{2,\infty}\mforall r<1/2,
\end{equation} and for which
\begin{equation}
  \label{eq:32}
  \lim_{t\to 0}\tfrac {\rad (u(t))}
{\sqrt t}=\sqrt 6.
\end{equation} This $u_0\notin \dot
B^{1/2,0}_{2,\infty}$  and hence $u_0\notin \dot H^{1/2}$  (note the
inclusions (\ref{eq:13})).

Note, in comparison with  (\ref{eq:32}),  that by  Theorem \ref{prop:ana_small_times1} for any $u_0\in
H^{1/2}$ the corresponding  solution obeys
\begin{equation}
  \label{eq:32H}
  \lim_{t\to 0}\tfrac {\rad (u(t))}
{\sqrt t}=\infty.
\end{equation}
  Moreover the
theory of Subsections \ref{Local solvability in dot r
and r for
large data} and \ref{Analyticity bounds in
    dot r and  r
    for small times and for large data} can be extended to the case of
   data
   $u_0\in\dot B^{1/2,0}_{2,\infty}$, cf. a  discussion in the
   beginning of the proof of Proposition \ref{prop:kato_small_data}. In particular for $u_0\in L^2\cap
   \dot B^{1/2,0}_{2,\infty}$ (\ref{eq:32H}) remains true for the
   corresponding solution. In fact the theory can be extended to the
   case  $u_0\in \dot B^{1/2}_{2,\infty}$ provided that  $\|u_0\|_{
     \dot B^{1/2}_{2,\infty}}$ is sufficiently small (so that
   (\ref{eq:6}) is fulfilled). This  leads to the
   existence of a
   unique real-analytic global solution, cf.   Subsection \ref{Global
     solvability in dot 1/2 for small data}. However in that case we can
   only conclude   weaker analyticity statements. To be specific, if  $u_0\in L^2\cap \dot
   B^{1/2}_{2,\infty}$ and $\|u_0\|_{
     \dot B^{1/2}_{2,\infty}}$ is sufficiently small we can conclude
   that $\liminf _{t\to 0}\big (\rad (u(t))/
\sqrt t\big )\geq \kappa$ for some $\kappa >0$ that depends  on
the (small)  norm $\|u_0\|_{
     \dot B^{1/2}_{2,\infty}}$. We have not calculated this norm  for
   the specific example given below, and consequently we do not know whether the example and therefore in particular
   (\ref{eq:32}) fit into this extended theory.

\subsubsection{An example}\label{subsubsec:example}
 Motivated by \cite{WJZHJ} we use the Hopf-Cole transformation \cite{Ev} and obtain a solution of the vector Burgers' equation: If $v=v(t,x)$ is  a positive solution to the
heat equation
\begin{equation}
  \label{eq:33}
  \frac{\partial }{\partial t}v=\triangle v;\ t>0,
\end{equation} then $w:=-2\ln v$ fulfills
\begin{equation}
  \label{eq:34}
 \frac{\partial }{\partial t}w+\tfrac12 |\nabla w|^2=\triangle w.
\end{equation} Taking first order partial derivatives in (\ref{eq:34})
we get
\begin{equation}
  \label{eq:34h}
 \frac{\partial }{\partial t}\partial_jw+(\nabla w\cdot \nabla
)\partial_jw=\triangle \partial_jw;\ j=1,2,3.
\end{equation}

Next take $M=I$ (in fact $M$ can be any invertible real $3\times
3$--matrix), $P=I$ and $u=\nabla
w$ (or more generally $u=M^{-1}\nabla
w$) we obtain from (\ref{eq:34h})
\begin{equation}
  \label{eq:35}
  \frac{\partial }{\partial t}u+(Mu\cdot \nabla)u-\triangle
  u=0\mand u=Pu.
\end{equation}
In particular we have constructed a solution to (\ref{eq:equamotion})
with $u_0=u(0,\cdot)$.

%This general construction will be used again in
%Subsection \ref{Example2}.

We choose
\begin{equation}
  \label{eq:36}
  v(t,x)=1-(t+1)^{-3/2}\exp\parb{-\tfrac {|x|^2}{4(t+1)}}.
\end{equation} Clearly (\ref{eq:33}) holds. We compute
\begin{equation}
  \label{eq:38}
  u_j(t,x)=\partial_jw(t,x)=-(t+1)^{-5/2}x_j\exp\parb{-\tfrac {|x|^2}{4(t+1)}}/v(t,x),
\end{equation} from which we obtain
\begin{equation}
  \label{eq:40}
  (u_0)_j(x)=\partial_jw_0(x)=-x_j\exp\parb{-\tfrac {|x|^2}{4}}/\parbb{1-\exp\parb{-\tfrac {|x|^2}{4}}}.
\end{equation} Here the denominator vanishes like ${|x|^2}/{4}$
at $x=0$. Consequently the components of $u_0$ have a Coulomb
singularity at  $x=0$.  In Fourier space this behaviour corresponds
to a decay like $\xi_j|\xi|^{-3}$ at infinity, cf. (\ref{eq:48})
stated below. Whence indeed $u_0\in L^2\cap\parb{\dot
B^{1/2}_{2,\infty}\setminus \dot B^{1/2,0}_{2,\infty}}$.

As for the property (\ref{eq:31}) we notice the (continuous)
embedding, cf. (\ref{eq:13}),
\begin{equation}
  \label{eq:44}
  L^2\cap \dot B^{1/2}_{2,\infty}\subseteq H^r\mfor  r\in [0,1/2[,
\end{equation} which by a scaling argument leads to the bound
\begin{equation}
  \label{eq:47}
  \|f\|_{\dot H^r}\leq C_r\|f\|_{L^2}^{1-2r}\|f\|^{2r}_{\dot
    B^{1/2}_{2,\infty}}\mfor  f\in L^2\cap \dot B^{1/2}_{2,\infty}\mand r\in [0,1/2[.
\end{equation}
We note the properties
\begin{align}
  \label{eq:45}
 &u\in BC([0,\infty[,L^2),\\
\label{eq:45B}
 &u\in
B([0,\infty[,\dot B^{1/2}_{2,\infty}).
\end{align}
  Only (\ref{eq:45B}), or equivalently
  \begin{equation}
    \label{eq:41}
    \sup_{t\geq0}\|u(t,\cdot)\|_{\dot B^{1/2}_{2,\infty}}<\infty,
  \end{equation} needs an elaboration.

 To prove (\ref{eq:41})  we introduce for $\kappa\geq1$ and $j=1,2,3$ the
  functions
  \begin{align*}
    f_j(\kappa,y)&=y_j\exp\parb{-\tfrac {|y|^2}{4}}/\parb{\kappa-\exp\parb{-\tfrac
        {|y|^2}{4}}},\\
\tilde  f_j(\kappa,y)&=4y_j/\parb{\tilde \kappa+|y|^2};\
\tilde \kappa=4(\kappa-1),
  \end{align*} and notice that
  \begin{equation*}
    u_j(t,x)=-(t+1)^{-1/2}f_j((t+1)^{3/2},x/\sqrt{t+1}).
  \end{equation*} We pick $\chi\in C_c^\infty (\R)$ with $\chi(s)=1$
  for $|s|<1$. By a scaling argument (\ref{eq:41}) will follow from
  the bound
\begin{equation}
    \label{eq:41b}
    \sup_{\kappa\geq 1}\|\chi(|\cdot|)f_j(\kappa,\cdot)\|_{\dot B^{1/2}_{2,\infty}}<\infty.  \end{equation} The proof of (\ref{eq:41b}) relies on a comparison
  argument. We notice that
\begin{equation}
    \label{eq:41c}
    \sup_{\kappa\geq 1}\|\chi(|\cdot|)\tilde  f_j(\kappa,\cdot)\|_{\dot B^{1/2}_{2,\infty}}<\infty,
  \end{equation} which may be seen as follows: First we notice the representation of the Fourier transform
  \begin{equation}\label{eq:48}
    (F \tilde f_j)(\kappa,\xi)=C\partial_{\xi_j}\big \{|\xi|^{-1}\int
    _0^\infty s^{-3/2}\exp\parb{-(4s)^{-1}}\e^{-\tilde
      \kappa|\xi|^2s}\d s\big \}.
  \end{equation} By computing the derivative and
  then estimating the second exponential  $\leq 1$ we deduce the
  bound $|(F\tilde  f_j)(\kappa,\xi)|\leq C |\xi|^{-2}$ uniformly in
  $\kappa\geq 1$. Using this estimate and the convoluton integral
  representation of the product
  we obtain that $|(F\{\chi\tilde  f_j\})(\kappa,\xi)|\leq C \inp{\xi}^{-2}$ uniformly in
  $\kappa\geq 1$ from which (\ref{eq:41c}) follows.

Due to (\ref{eq:41c}) and (\ref{eq:13}) it suffices for (\ref{eq:41b}) to show
\begin{equation}
    \label{eq:41d}
    \sup_{\kappa\geq 1}\|\chi(|\cdot|)\{f_j(\kappa,\cdot)-\tilde f_j(\kappa,\cdot)\}\|_{H^{1}}<\infty.
  \end{equation} Clearly (\ref{eq:41d}) follows from the uniform
  pointwise bounds
  \begin{equation*}
   |\chi(|y|)\{f_j(\kappa,y)-\tilde
  f_j(\kappa,y)\}|\leq C|y|\mand  |\nabla \parb{\chi(|y|)\{f_j(\kappa,y)-\tilde
  f_j(\kappa,y)\}}|\leq C,
  \end{equation*}
 which in turn follow  from elementary Taylor expansion.

 We
conclude from (\ref{eq:47})--(\ref{eq:45B})  that indeed
\begin{equation}
  \label{eq:39}
  u\in BC([0,\infty[,H^r)\mfor  r\in [0,1/2[.
\end{equation}

As for the property (\ref{eq:32}) we claim more generally that
\begin{equation}
  \label{eq:49}
  \rad(u(t))^2 = 6(t+1)\ln(t+1)\mforall t> 0.
\end{equation}
  To see this note that $v(t,iy)$ is zero on the surface $|y|^2 = 6(t+1)\ln(t+1)$ and if $|y|^2 < (1-\epsilon)6(t+1)\ln(t+1)$ then $|v(t,x+iy)| > 1-(t+1)^{-3\epsilon/2}$.

%\subsection{Analyticity in small time variable}\label{Analyticity in
%  small time variable for large data}

\section{Analyticity bounds for all times}
In this section we shall study analyticity properties of the
global small data solutions of Proposition \ref{prop:kato big data}.
\subsection{Global analyticity bounds in $\dot H^{1/2}$  and  $H^{1/2}$
     for small  data} \label{Analyticity bounds in
    dot 1/2 and  1/2
    for all  times and for small  data}
%\subsection{Global analyticity in time variable for small  data}\label{Analyticity in time variable for small  data}
Let us begin this subsection by considering  $r\in [1/2, 3/2[$ as in
Subsection \ref{Analyticity bounds in
    dot r and  r
    for small times and for large data}. The contraction condition
  (\ref{eq:6}) leads to the following combination of (\ref{eq:18}) and
  (\ref{eq:8bba})
\begin{equation}
  \label{eq:25bb}
  C_4\inp{\lambda}^{5/2-r}\e^{\lambda^2/2}T^{r/2-1/4}<1;\
  C_4=4C_1
  C_3\|A^ru_0\|.
\end{equation} The constants $C_1$ and $ C_3$ from
(\ref{eq:18}) and (\ref{eq:8bba}), respectively, are independent of
$\lambda$, $T$ and $u_0$ (but depend on $r$).

Obviously (\ref{eq:25bb}) cannot be fulfilled for $T=\infty$ unless
$r=1/2$. On the other hand if $r=1/2$ and $\|A^ru_0\|$ is sufficiently
small the condition is fulfilled for  $T=\infty$ for $\lambda\geq 0$
smaller than some critical positive number. This observation leads to
the following global analyticity result (using again Lemma
\ref{Lemma:spect_bnd}, (\ref{eq:7cq}) and (\ref{eq:7bbbq})):

\begin{prop}\label{prop:kato big data anal} Suppose $u_0\in\dot
H^{1/2}$ and that the constant $C_4=4C_1
  C_3\|A^{1/2}u_0\|$ in (\ref{eq:25bb}) (with $r=1/2$) obeys
  $C_4<1$. For $u_0\neq 0$ define $\bar \lambda>0$ as the solution to  the equation
  \begin{equation*}
   4C_1
  C_3\|A^{1/2}u_0\|\inp{\bar \lambda}^{2}\e^{\bar \lambda^2/2}=1.
  \end{equation*} If $u_0= 0$ define $\bar \lambda=\infty$.

  Then the solution $X$ to the
integral equation (\ref{eq:5})  as constructed in Proposition
\ref{prop:kato big data} obeys the following bounds uniformly in
$\lambda\in[0,\bar \lambda[$ and $t>0$
\begin{subequations}
  \begin{align}
  \label{eq:50}
  \|A^{5/4}\e^{\lambda\sqrt tA}X(t)\|&\leq
  2C_1\|A^{1/2}u_0\|\inp{\lambda}^{3/4}\e^{\lambda^2/4}t^{-3/8},\\
\|A^{1/2}\e^{\lambda\sqrt tA}X(t)\|&\leq
  C\parbb{\e^{\lambda^2/4}\|A^{1/2}u_0\|+\inp{\lambda}^2\e^{3\lambda^2/4}\parb{2C_1\|A^{1/2}u_0\|}^2},\label{eq:50b}
\end{align} and if in addition $u_0\in L^2$,
\begin{equation}
  \label{eq:52}
  \|\e^{\lambda\sqrt tA}X(t)\|\leq
  \widetilde C\parbb{\e^{\lambda^2/4}\|u_0\|+\inp{\lambda}^{3/2}\e^{3\lambda^2/4}\parb{2C_1\|A^{1/2}u_0\|}^2t^{1/4}}.
\end{equation}
\end{subequations}
\end{prop}

\subsection{Improved global analyticity bounds in $\dot H^{1/2}$  and  $H^{1/2}$
    for small  data} \label{Improved analyticity
    bounds in dot 1/2 and  1/2
    for all  times and for small  data}
In this subsection we improve on the results of Subsection \ref{Analyticity bounds in
    dot 1/2 and  1/2
    for all  times and for small  data} along the line of the method
  of Subsection \ref{Improved analyticity
    bounds in dot r and  r, r>
  1/2,
    for small times and for large data}. Whence we fix $u_0\in\dot
H^{1/2}$ and $T\in]0,\infty[$ and define the underlying Banach space $\vB$ in terms of
 the finite interval $I=]0,T]$ and a
time-dependent choice of $\lambda$ (and (\ref{eq:20}) with $\tilde
r-r=3/4$, and the parameters  $s_1=3/8$ and $s_2=5/4$). Explicitly we choose
$\lambda=\lambda(t)$ as in (\ref{eq:21}) to be used in the expression
(\ref{eq:20}) (with $\tilde r-r=3/4$). As in the previous subsection
we will need the parameter $\lambda_0$ of (\ref{eq:21}) to be smaller
than a certain critical positive number dictated by the  contraction
condition (\ref{eq:25}) (with $r=1/2$). That is we need
\begin{equation}
  \label{eq:25bbc}
  C_4\inp{\lambda_0}^{2}\e^{\lambda_0^2/4}<1;\
  C_4=4C_1
  C_3\|A^{1/2}u_0\|;
\end{equation} here the  constants $C_1$ and $ C_3$ from
(\ref{eq:18}) and (\ref{eq:24}), respectively, are (again) independent of
$\lambda$, $T$ and $u_0$. (Notice that (\ref{eq:25bbc}) ``improves''
(\ref{eq:25bb}) (for $r=1/2$)  in that the exponent $\lambda^2/2\to
\lambda_0^2/4$.) Mimicking the proofs of Theorem
\ref{thm:impr_ana_small_times1} and Corollary \ref{cor:small} we
obtain under the condition (\ref{eq:25bbc}) the following improvement
of Proposition \ref{prop:kato big data anal}:
\begin{thm}
\label{thm:global anal}
  \begin{enumerate}[i)]
  \item \label{item:12}
Suppose $u_0\in\dot
H^{1/2}$ and that the constant $C_4=4C_1
  C_3\|A^{1/2}u_0\|$ in (\ref{eq:25bbc})  obeys
  $C_4<1$. For $u_0\neq 0$ define $\bar \lambda>0$ as the solution to  the equation
  \begin{equation}\label{eq:51}
   4C_1
  C_3\|A^{1/2}u_0\|\inp{\bar \lambda}^{2}\e^{\bar \lambda^2/4}=1.
  \end{equation} If $u_0= 0$ define $\bar \lambda=\infty$.

Then the solution $X$ to the integral equation (\ref{eq:5}) as constructed in Proposition \ref{prop:kato big data} obeys the following bounds uniformly in
$\lambda_0\in[0,\bar \lambda[$ and $t>0$
\begin{subequations}
 \begin{align}
  \label{eq:50c}
  \|A^{5/4}\e^{\lambda_0\sqrt tA}X(t)\|&\leq
  2C_1\|A^{1/2}u_0\|\inp{\lambda_0}^{3/4}\e^{\lambda_0^2/4}t^{-3/8},\\
\|A^{1/2}\e^{\lambda_0\sqrt tA}X(t)\|&\leq
  C\parbb{\e^{\lambda_0^2/4}\|A^{1/2}u_0\|+\inp{\lambda_0}^2\e^{\lambda_0^2/2}\parb{2C_1\|A^{1/2}u_0\|}^2}.\label{eq:50bc}
\end{align}

\item \label{item:12b} Suppose  in addition that $u_0\in L^2$. Then
\begin{equation}
  \label{eq:52c}
  \|\e^{\lambda_0\sqrt tA}X(t)\|\leq
  \widetilde C\parbb{\e^{\lambda_0^2/4}\|u_0\|+\inp{\lambda_0}^{3/2}\e^{\lambda_0^2/2}\parb{2C_1\|A^{1/2}u_0\|}^2t^{1/4}}.
\end{equation}
\end{subequations}
\end{enumerate}
\end{thm}
We will use the following corollary in Subsection \ref{Completing the proof}.  We omit its proof.
\begin{cor}\label{cor:o(1)}
Suppose $u_0\in\dot
H^{1/2}$ and that the constant $C_4=4C_1
  C_3\|A^{1/2}u_0\|$ in (\ref{eq:25bbc})  obeys
  $C_4<1$.  Suppose $ \liminf_{t \to \infty}\|A^{1/2}X(t)\| = 0$ where $X$ is the solution to the integral equation (\ref{eq:5}) as constructed in Proposition \ref{prop:kato big data}. Then for any $\lambda \geq 0$, as $t \to \infty$
\begin{subequations}
 \begin{align}
t^{3/8}\|A^{5/4}\e^{\lambda \sqrt tA}X(t)\| = o(1),\\
\|A^{1/2}\e^{\lambda \sqrt tA}X(t)\| = o(1).
\end{align}
\end{subequations}
\end{cor}

\begin{remark}
  \label{remark:global anal} If  $u_0\in
 H^{1/2}$ and $\bar \lambda>0$ obeys (\ref{eq:51}) then clearly $\bar
\lambda \sqrt t$ is a lower bound of $\rad (X(t))$. In particular
in the sense of taking  the limit
$\|A^{1/2}u_0\|\rightarrow 0$ we obtain
\begin{equation}
  \label{eq:53}
  \liminf \tfrac {\rad (X(t))}{\sqrt{-4\ln \|A^{1/2}u_0\|}\sqrt
  t}\geq 1\text{ uniformly in } t>0.
\end{equation}
 Clearly this statement is weak  in the
 small time regime compared to (\ref{eq:32H}). On the other hand,
  as demonstrated in Section \ref{Optimal rate of growth of analyticity
    radius },  (\ref{eq:53}) is useful for obtaining
 further improved bounds of  the  analyticity radius in
 the large time regime.
\end{remark}

\section{Differential inequalities for small global solutions in $\dot
  H^{1/2}$  and  $H^{1/2}$} \label{Differential inequalities for small
  global solutions in dot 1/2  and  1/2}
In this section we continue our  study of  analyticity bounds of global small
data   solutions in $\dot
  H^{1/2}$  and  $H^{1/2}$ initiated in the previous section.  This
  study will be continued and completed in Section \ref{Optimal rate of growth of analyticity
    radius } where some optimal analyticity radius bounds in
 the large time regime will be presented. We  impose throughout this
 section the
 conditions of Theorem \ref{thm:global anal} \ref{item:12}. Notice
 that a simplified version of the bound (\ref{eq:50bc})
  takes the form
  \begin{equation}
    \label{eq:55}
    \|A^{1/2}\e^{\lambda_0\sqrt tA}X(t)\|\leq
  \breve C\inp{\bar \lambda}^{-2}\mforall \lambda_0\in [0,\bar
  \lambda[\mand t>0.
  \end{equation}  Our goal is twofold:
\begin{enumerate}[1)]
\item\label{item:5i}  Under an additional decay condition of the quantity $\|A^{1/2}X(t)\|$ we shall  improve on the right hand side of (\ref{eq:55}) in
 the large time regime. A similar improvement of (\ref{eq:52c}) will
  be established  in terms of decay of the quantity $\|X(t)\|$.
\item\label{item:6i}  Under an additional decay condition on the
  quantity $\|X(t)\|$ we shall
  show decay of the quantity $\|A^{1/2}X(t)\|$.
\end{enumerate}
 We think \ref{item:5i} has some independent interest, although more
 refined bounds will be presented in Section \ref{Optimal rate of growth of analyticity
    radius } (in particular presumably better
 bounds on analyticity radii than can be derived from the methods
 presented here). As for
 \ref{item:6i}, our  result will be used in Section \ref{Optimal rate of growth of analyticity
    radius }.

The  analysis is partly inspired by \cite{FT}, \cite{Sc1} and \cite{OT}.

\subsection{Energy inequality} \label{Energy inequality} Partly as a
motivation we recall here a
version of the
 energy inequality well-known for a class of solutions to
 (\ref{eq:equamotiona}). We  state it for the  function $X$ of
 Theorem \ref{thm:global anal} \ref{item:12} subject to the further conditions
  \begin{equation}
    \label{eq:54}
    u_0=Pu_0\in L^2\mand \nabla\cdot \parb{MX(t)}=0\mforall t>0.
  \end{equation} Notice that the second condition of (\ref{eq:54}) is
  fulfilled for the problem (\ref{eq:equamotiona}). Using Theorem \ref{thm:global anal},
  (\ref{eq:equamotion}) and (\ref{eq:54}) we can derive
  \begin{equation}
    \label{eq:56}
   \tfrac{\d }{\d t}\|X(t)\|^2=-2\|AX(t)\|^2 +2\inp[\big]{X(t),(MX(t)\cdot \nabla)X(t)}=-2\|AX(t)\|^2\mforall t>0.
  \end{equation} We refer the
   reader to Subsection \ref{Global
     solutions} for a  discussion relevant for this
   derivation. (Actually a more general result than (\ref{eq:56}) is
   stated in  Corollary
   \ref{cor:energybound2}  in Subsection \ref{Sobolev and analiticity  bounds for finite
  intervals}.) We obtain by integrating (\ref{eq:56})
  \begin{equation}
    \label{eq:57}
    \|X(t)\|^2=\|u_0\|^2-2\int_0^t\|AX(s)\|^2\,\d s\mforall t>0.
  \end{equation} In particular the {\it energy inequality }$
  \|X(t)\|^2\leq \|u_0\|^2$ holds. Clearly this bound improves the
  bound $ \|X(t)\|=O(t^{1/4})$ of (\ref{eq:52}) at
  infinity.

We notice that for the problem
(\ref{eq:equamotiona})  and under the above conditions
 it can be proven  that  the quantity $ \|X(t)\|=o(t^0)$  as $t\to \infty$. Under some further (partly generic)
conditions it is shown in \cite{Sc1} that  $ \|X(t)\|=O(t^{-5/4})$
while $\|X(t)\|\neq o(t^{-5/4})$.

\subsection{Differential inequalities for exponentially weighted Sobolev norms}\label{Differential inequalities for analyticity radius growing
  like sqrt t}
Under the conditions of
 Theorem \ref{thm:global anal} \ref{item:12}
introduce for $r\geq 0$ and $\lambda_0\in ]0,\bar \lambda[$ the quantities
\begin{equation}
  \label{eq:58}
  J_r(t)=\|A^rX(t)\|^2\mand G_r(t)=\|A^r\e^{\lambda_0\sqrt
    tA}X(t)\|^2\mfor t>0.
\end{equation} We are mainly interested in these quantities for
$r=1/2$ and $r=0$. Any consideration for $r\in [0, 1/2[$ will involve the
additional requirement $u_0\in L^2$. Under the additional conditions (\ref{eq:54}) we have,
due to the previous subsection,
the a priori bound $J_{0}(t)=O(t^{-\sigma})$  with  $\sigma=0$ for
   $t\to \infty$ however in the following we do not assume (\ref{eq:54}).
\begin{lemma}
  \label{lemma:dif-eq} For all  $\kappa>0$
   and $\lambda_0\in ]0,\bar \lambda[$ there exists
  $K=K(\kappa, \lambda_0)>0$
  such that
  \begin{equation}
    \label{eq:59}
    \tfrac{\d }{\d t}G_{1/2}(t)\leq -\kappa
    t^{-1}G_{1/2}(t)+Kt^{-1}J_{1/2}(t) \mforall t>0.
  \end{equation}
\end{lemma}
\begin{proof} We compute
  \begin{align}
    \label{eq:60}
  \tfrac{\d }{\d t}G_{1/2}(t)&=-2  G_{3/2}(t)+\lambda_0 t^{-1/2}G_{1}(t)+R(t);\\R(t)&=2\inp{A^{1/2}\e^{\lambda_0\sqrt
    tA}X(t),A^{1/2}\e^{\lambda_0\sqrt
    tA}P(MX(t)\cdot \nabla)X(t)}.\nonumber
  \end{align}
By using the Cauchy-Schwarz inequality and Lemma \ref{lemma:sobolev}
we are led to the bounds
\begin{equation}
  \label{eq:61}
  R(t)\leq  2G_{1}(t)^{1/2} CG_{1}(t)^{1/2} G_{3/2}(t)^{1/2} \leq
  G_{3/2}(t) +C^2G_1(t)^2 \mforall t>0.
\end{equation}

 Now, pick any $\lambda_1\in ]\lambda_0, \bar \lambda[$. We can
 estimate the second term on the right hand side of (\ref{eq:61}) by
 first
 using the  Cauchy-Schwarz inequality and (\ref{eq:55}) (with
 $\lambda_0\to \lambda_1$) to obtain
 \begin{align}
   \label{eq:62}
  G_{1}(t)^2\leq G_{3/2}(t)G_{1/2}(t)&\leq  \tfrac {\sup_{x\geq
      0}x^2\e^{-2x}}{(\lambda_1-\lambda_0)^2t}G_{1/2,\lambda_1}(t)
  G_{1/2}(t) \leq \widetilde Ct^{-1}G_{1/2}(t);\\&\;\widetilde C=C(\lambda_0)\breve C\inp{\bar \lambda}^{-2}.\nonumber
 \end{align} Clearly (\ref{eq:60})--(\ref{eq:62}) lead to
 \begin{equation}
   \label{eq:63}
   \tfrac{\d }{\d t}G_{1/2}(t)\leq -  G_{3/2}(t)+\lambda_0
   t^{-1/2}G_{1}(t)+C^2\widetilde C t^{-1}G_{1/2}(t).
 \end{equation}

Next we insert $0= \kappa
    t^{-1}G_{1/2}(t) -Kt^{-1}J_{1/2}(t)-\kappa
    t^{-1}G_{1/2}(t) +Kt^{-1}J_{1/2}(t)$ on the right hand
side of (\ref{eq:63}). We need to examine the condition
\begin{equation}
  \label{eq:64}
  -  G_{3/2}(t)+\lambda_0 t^{-1/2}G_{1}(t) +(\kappa+C^2\widetilde C )
    t^{-1}G_{1/2}(t)-Kt^{-1}J_{1/2}(t)\leq0.
\end{equation} By the spectral theorem the
bound (\ref{eq:64}) will follow from
\begin{equation}
  \label{eq:65}
  -x^3\e^{2\lambda_0 x}+\lambda_0 x^2\e^{2\lambda_0 x}+(\kappa+C^2\widetilde C ) x\e^{2\lambda_0 x}\leq Kx\mforall x\geq 0.
\end{equation} The estimate (\ref{eq:65}) is obviously fulfilled for
some
 $K=K(\kappa,\lambda_0)>0$.
\end{proof}

 \begin{cor}
   \label{cor:decay1}
   \begin{enumerate}[i)]
   \item \label{item:9} For all  $\kappa>0$
  and  $\lambda_0\in ]0,\bar \lambda[$ there exists
  $K=K(\kappa, \lambda_0)>0$
  such that
   \begin{equation}
     \label{eq:66}
    G_{1/2}(t)\leq Kt^{-\kappa}\int _0^ts^{\kappa-1}J_{1/2}(s)\,\d s
    \mforall t>0.
   \end{equation}

 \item \label{item:10} Suppose that  for some $\sigma> -1/2$  the
   bound $J_{1/2}(t)=O(t^{-\sigma-1/2})$ for
   $t\to \infty$ holds. Then $J_{r}(t)=O(t^{-\sigma-r})$ and
   $G_{r}(t)=O(t^{-\sigma-r})$ for all $r\geq 1/2$.

 \item \label{item:11} Suppose the conditions of
 Theorem \ref{thm:global anal} \ref{item:12b} and that    for some $\sigma> -1/2$  the
   bound
$J_{0}(t)=O(t^{-\sigma})$
holds. Then
  $J_{r}(t)=O(t^{-\sigma-r})$ and $G_{r}(t)=O(t^{-\sigma-r})$ for all  $r\geq 0$.
  \item\label{item:newitem}Suppose the conditions of
 Theorem \ref{thm:global anal} \ref{item:12b} and $J_0(t) = O(1)$.  Then $G_0(t) = o(1)$.
   \end{enumerate}
\end{cor}
\begin{proof} As for \ref{item:9} notice that  $t^\kappa$ is
  an integrating factor for (\ref{eq:59}).

For \ref{item:10} we choose
  $\kappa> 1/2+\sigma$ in the bound (\ref{eq:66}) yielding the
   bound $G_{1/2}(t)=O(t^{-\sigma-1/2})$. Whence also
   $G_{r}(t)=O(t^{-\sigma-r})$ for all $r\geq 1/2$ (here we used the
   quantity $G_{1/2,\lambda_1}(t)$ of (\ref{eq:62})). In particular
   $J_{r}(t)=O(t^{-\sigma-r})$ for all $r\geq 1/2$.

For \ref{item:11} we first prove that $G_{1/2}(t)=O(t^{-\sigma-1/2})$.
By the Cauchy-Schwarz inequality
  \begin{equation}
    \label{eq:67}
    J_{1/2}(t)\leq \sup_{x\geq 0}x^{1/2}\e^{-\lambda_0 \sqrt t x}\,J_{0}(t)^{1/2}
    G_{1/2}(t)^{1/2}\leq \tfrac {C}{\lambda_0\sqrt {t} }KJ_{0}(t) + G_{1/2}(t)/K.
  \end{equation}
In combination with (\ref{eq:59}) this estimate  leads to
\begin{equation}
    \label{eq:59b}
    \tfrac{\d }{\d t}G_{1/2}(t)\leq -(\kappa-1)
    t^{-1}G_{1/2}(t)+C(\kappa,\lambda_0)t^{-3/2}J_{0}(t);\;C(\kappa,\lambda_0)=K^2\tfrac {C}{\lambda_0}.
  \end{equation}
By choosing   $\kappa> 2+\sigma$ in (\ref{eq:59b})  we deduce
 the following analogue
of (\ref{eq:66}) where $\widetilde K:=C(\kappa,\lambda_0)$ and $\tilde \kappa:=\kappa-1$
\begin{equation}
     \label{eq:66b}
    G_{1/2}(t)\leq \widetilde Kt^{-\tilde \kappa}\int _0^ts^{\tilde\kappa-3/2}J_{0}(s)\,\d s
    \mforall t>0.
   \end{equation} It follows from (\ref{eq:66b}) that indeed  $G_{1/2}(t)=O(t^{-\sigma-1/2})$.

To complete the proof of \ref{item:11} it suffices to
 show that $G_0(t)=O(t^{-\sigma})$. Split
\begin{equation*}
  G_0(t)=\|1_{[0,1[}(\sqrt
    tA)\e^{\lambda_0\sqrt
    tA}X(t)\|^2+\|1_{[1,\infty[}(\sqrt
    tA)\e^{\lambda_0\sqrt
    tA}X(t)\|^2.
\end{equation*} The first term bounded by
$\e^{2\lambda_0}J_{0}(t)=O(t^{-\sigma})$. The second term is bounded by
\begin{equation*}
  \|\big (\sqrt tA\big )^{1/2}\e^{\lambda_0\sqrt
    tA}X(t)\|^2= t^{1/2}G_{1/2}(t)=t^{1/2}O(t^{-\sigma-1/2})=O(t^{-\sigma}).
\end{equation*}

To prove \ref{item:newitem} we use the integral equation (\ref{eq:1}) in the form $X = Y + B(X,X)$ where
$Y(t) = \e^{-tA^2}u_0$. The decay of the first term is clear.  For the nonlinear term we split the integral from $0$ to $t$ into an integral from $0$ to $T$ and another from $T$ to $t$.  Using the bound on $G_{5/4}$ from \ref{item:11} in the second integral we obtain a term which is $O(T^{-1/4})$ while the first integral is given by $\e^{-(t-T)A^2}g(T)$ for some $g(T) \in L^2$ and is thus $o(1)$.
\end{proof}
\begin{remark*}
   In studying the types of inequalities proved in Corollary \ref{cor:decay1} we were motivated by [OT].  However the main theorem in that paper is stated incorrectly and the proof given there is also incorrect.
\end{remark*}
For completeness of presentation we  end this section by giving another proof of  Corollary
\ref{cor:decay1} \ref{item:11}. Although the proof  goes along similar  lines
it is somewhat more direct.
\begin{thm}
  \label{thm:second} Suppose  that
$u_0\in H^{1/2}$. For all $\kappa>0$  and $\lambda_0\in
]0,\bar \lambda[$ there exists $K=K(\kappa,\lambda_0)>0$
  such that
  \begin{equation}
    \label{eq:59cc}
    \tfrac{\d }{\d t}G_{0}(t)\leq -\kappa
    t^{-1}G_{0}(t)+Kt^{-1}J_{0}(t) \mforall t>0.
  \end{equation}
Whence \begin{equation}
     \label{eq:66cc}
    G_{0}(t)\leq Kt^{-\kappa}\int _0^ts^{\kappa-1}J_{0}(s)\,\d s
    \mforall t>0.
   \end{equation}

 In particular if for some $\sigma> -1/2$   the  bound
$J_{0}(t)=O(t^{-\sigma})$
holds,  then $G_{0}(t)=O(t^{-\sigma})$  and whence, more
generally, $J_{r}(t)=O(t^{-\sigma-r})$ and $G_{r}(t)=O(t^{-\sigma-r})$ for all $r\geq 0$.
\end{thm}
\begin{proof} We
  compute
\begin{equation}
    \label{eq:60c}
  \tfrac{\d }{\d t}G_{0}(t)=-2  G_{1}(t)+\lambda_0 t^{-1/2}G_{1/2}(t)+R(t),
  \end{equation} where
\begin{equation}
  \label{eq:61c}
  R(t)\leq  2G_{0}(t)^{1/2} CG_{1}(t)^{1/2} G_{3/2}(t)^{1/2} \leq G_{1}(t)+
  C^2G_{0}(t) G_{3/2}(t).
\end{equation}
   In particular, due to the estimate
   $G_{3/2}(t)=G_{3/2,\lambda_0}(t)\leq C(\lambda_1 -\lambda_0
  )t^{-1}G_{1/2,\lambda_1}(t)$,
   $\lambda_1\in ]\lambda_0,\bar \lambda[$, and  (\ref{eq:55})
   applied to $G_{1/2,\lambda_1}(t)$,
\begin{equation}
  \label{eq:61cc}
  R(t)\leq  G_{1}(t)+
  \widetilde C t^{-1}G_{0}(t).
\end{equation}

To obtain (\ref{eq:59cc}) we insert $0= \kappa
    t^{-1}G_{0}(t) -Kt^{-1}J_{0}(t)-\kappa
    t^{-1}G_{0}(t) +Kt^{-1}J_{0}(t)$ on the right hand
side of (\ref{eq:60c}), and due to (\ref{eq:61cc}) we need only to examine the condition
\begin{equation}
  \label{eq:64cc}
  -  G_{1}(t)+\lambda_0 t^{-1/2}G_{1/2}(t) +\widetilde C t^{-1}G_{0}(t)+\kappa
    t^{-1}G_{0}(t)-Kt^{-1}J_{0}(t)\leq0.
\end{equation} By the spectral theorem the
bound (\ref{eq:64cc}) will follow from
\begin{equation}
  \label{eq:65cc}
  -x^2\e^{2\lambda_0 x}+\lambda_0 x\e^{2\lambda_0 x}+(\kappa +\widetilde C )\e^{2\lambda_0 x}\leq K\mforall x\geq 0,
\end{equation} which in turn   obviously  is valid for some
$K=K(\kappa,\lambda_0)$. Whence we have shown (\ref{eq:59cc}).

The   remaining
statements are  immediate consequences of (\ref{eq:59cc}), cf. the
proof of Corollary \ref{cor:decay1}.
\end{proof}

\section{Optimal rate of growth of  analyticity
  radii }\label{Optimal rate of growth of analyticity
  radius }
We shall combine Subsections \ref{Improved analyticity
    bounds in dot 1/2 and  1/2
    for all  times and for small  data} and \ref{Differential inequalities for analyticity radius growing
  like sqrt t} to obtain improved analyticity radius bounds of  global solutions
with small data in  $\dot H^{1/2}$  or   $H^{1/2}$
in the large time regime. An example shows that our bounds are optimal.

\subsection{Optimizing  bounds of  analyticity
  radii for large times} \label{Optimizing the bound of the analyticity
  radius for large times}
Suppose the conditions of Theorem \ref{thm:global anal} \ref{item:12} and that for some $\sigma> -1/2$  the
   bound $\|A^{1/2}X(t)\|=O(t^{-(2\sigma+1)/4})$ for
   $t\to \infty$ holds. We shall then apply Theorem \ref{thm:global anal}
   to $u_0\to X(T)$ and $X\to u_{T}$, where
   $u_{T}(\tau):=X(\tau+T)$; here  $T>0$ is an auxiliary  variable
   that in the end will be large (proportional to the time $t=\tau+T$). Notice
   that $X(T)\in \dot H^{1/2}$ (since we have assumed  that
   $u_0\in \dot H^{1/2}$), and that $u_{T}$ is  the unique
   small solution to
   the
   integral equation (\ref{eq:5}) with data $X(T)$, cf.  Proposition
   \ref{prop:kato big data}. The fact that here indeed $u_{T}$ is  a  solution to
   (\ref{eq:5}) requires an argument not given here. (We refer the
   reader to Subsection \ref{Global
     solutions} for a thorough discussion of related issues in a
   different setting.)
%alternatively  we may invoke
%   Remarks \ref{rmks:strong-solutions}  \ref{item:15} and \ref{item:17}
%   and Proposition \ref{prop:strong-solutions}.
Since
   $T^{-(2\sigma+1)/4}\to 0$ for $T\to \infty$ we obtain for the
   critical value, $\bar \lambda=\bar \lambda(T)$ of (\ref{eq:51}), that
   $\bar \lambda\to \infty$ for $T\to \infty$. In fact
   \begin{equation*}
     \liminf _{T\to \infty}\bar \lambda/\sqrt {\ln(T)}\geq (2\sigma+1)^{1/2}.
   \end{equation*} Consequently for any $\epsilon_0\in]0,1[$, $\lambda_0:=\sqrt
   {(2\sigma+1)(1-\epsilon_0)\ln(T)}$ is a legitimate choice in
   Theorem \ref{thm:global anal} with $u_0\to X(T)$  provided  $T$
   is large enough.
We shall use this observation to prove the following main result.
\begin{thm}\label{thm:opt}
  Suppose the conditions of Theorem \ref{thm:global anal} \ref{item:12}, i.e. $u_0\in\dot
H^{1/2}$ and that the constant $4C_1
  C_3\|A^{1/2}u_0\|<1$. Let $X$ denote the corresponding solution to
  the integral equation (\ref{eq:5}).  We have:
  \begin{enumerate}[i)]
  \item \label{item:7} Suppose   that for some $\sigma> -1/2$  the
  following bound holds
  \begin{equation}
    \label{eq:69a}
   \|A^{1/2}X(t)\|=O(t^{-(2\sigma+1)/4})\mfor t\to \infty.
  \end{equation}
     Let $0\leq \tilde
   \epsilon<\epsilon\leq 1$ be given. Then there exist constants
   $t_0>1$ and $C>0$  such that
   \begin{equation}
     \label{eq:68}
     \|A^{1/2}\exp\parbb {\sqrt{(1-\epsilon)(2\sigma+1)}\sqrt{t\ln t}A} X(t)\|\leq C t^{-\tilde \epsilon (2\sigma+1)/4}\mforall
   t\geq t_0.
   \end{equation}
\item \label{item:8}
Suppose $u_0\in L^2$, and  that for some $\sigma> -1/2$  the
  following bound holds
\begin{equation}
    \label{eq:69b}
    \|X(t)\|=O(t^{-\sigma/2})\mfor t\to \infty.
  \end{equation}
Then  (\ref{eq:69a}) holds (and therefore in particular the conclusion
of \ref{item:7}).

Let $0\leq \tilde
   \epsilon<\epsilon\leq 1$ be given. Then there exist constants
   $t_0>1$ and $C>0$  such that
   \begin{equation}
     \label{eq:68b}
     \|\exp\parbb {\sqrt{(1-\epsilon)(2\sigma+1)}\sqrt{t\ln t}A} X(t)\|\leq C t^{1/4-\tilde \epsilon (2\sigma+1)/4}\mforall
   t\geq t_0.
   \end{equation}

In particular
 \begin{equation}
    \label{eq:46ww}
 \liminf_{t\to \infty}\tfrac{\rad(X(t))}{\sqrt{t\ln t}}\geq \sqrt{2\sigma+1}.
  \end{equation}
\end{enumerate}
\end{thm}
\begin{proof}
  We prove first \ref{item:7}. So fix $0\leq \tilde
   \epsilon<\epsilon\leq 1$. To make contact to the discussion at the
   beginning of this subsection let us then choose $\epsilon_0\in ]\tilde
   \epsilon,\epsilon[$. We introduce in addition to the variable
   $T$  a ``new time''  $\tau$  and a parameter  $n$ by the relations
   \begin{equation}\label{eq:70}
    t=\tau+T=(n+1)T;\;T\geq T_0.
   \end{equation} We will let the parameters $n$ and  $T_0$ be chosen
   large. First we fix  $n$ by  the condition
   \begin{equation}\label{eq:69}
   n(1-\epsilon_0)> (n+1)(1-\epsilon).
   \end{equation}

As noted at the
   beginning of this subsection  we are allowed to choose
\begin{equation}
  \label{eq:71}
 \lambda_0=\sqrt
   {(2\sigma+1)(1-\epsilon_0)\ln(T)}
\end{equation} in
   Theorem \ref{thm:global anal} with $u_0\to X(T)$  provided  $T\geq
   T_0$ for some large $T_0>0$. We do that, and estimate using
   (\ref{eq:70}) and (\ref{eq:69})
   \begin{align}
     \label{eq:72}
    \lambda_0\sqrt{\tau}&= \sqrt
   {(2\sigma+1)(1-\epsilon_0)\ln(t/(n+1))} \sqrt{tn/(n+1)}\nonumber \\ & \geq
   \sqrt {(2\sigma+1)(1-\epsilon)\ln(t)} \sqrt{t};\mforall\;t\geq t_0:=
   (n+1)T_0\text{ for
     a }T_0>0.
   \end{align}  Here $T_0>0$ possibly needs to be chosen larger than
   before. Now fix such a $T_0$. Whence also
     $t_0$ is fixed, and with this value of  $t_0$ indeed the left
     hand side of (\ref{eq:68}) is finite for all $t\geq t_0$. The
     bound (\ref{eq:68}) follows then from (\ref{eq:50bc}). We have proved
     \ref{item:7}.

As for \ref{item:8}, the first statement is a consequence of Corollary
\ref{cor:decay1} \ref{item:11}.
The second statement  follows from the  proof of  \ref{item:7} and (\ref{eq:52c}).
\end{proof}

\subsection{Example} \label{Example2} For the example presented in
Subsection \ref{Example1} we have  the conditions of Theorem
\ref{thm:opt} \ref{item:8} fulfilled after a translation in
time of the given solution $u$; i.e. by replacing $u\to u_{\widetilde T}$ for a sufficiently large
$\widetilde T>0$. This is with $\sigma=5/2$. The estimate
(\ref{eq:46ww}) is {\it sharp}
by (\ref{eq:49}). Similarly the more precise bounds  (\ref{eq:68b})
are
 sharp. More precisely  the power of $t$ on  the right of (\ref{eq:68b}) cannot be
improved for any $\epsilon\in ]0,1[$ since indeed  the estimate is false with  $\tilde\epsilon=\epsilon$. This
follows readily from  an examination of the analytic extension of
(\ref{eq:38}). Likewise (\ref{eq:68})  is
sharp in the same sense. This can be seen by using  the
optimality of (\ref{eq:68b}) discussed above and an argument similar
to the one presented at the end of the proof of Corollary
\ref{cor:decay1} \ref{item:11}.

 \section{Global solutions for arbitrary data}\label{Openness of the set of global solutions with respect to data}
In this section we shall discuss strong  solutions and in
particular strong global solutions without assuming the
$\dot H^{1/2}$-smallness condition of Subsection \ref{Global solvability in dot 1/2 for small data}.
\subsection{Strong solutions} \label{Global solutions} Dealing with  global
large data solutions we need first to
define a  notion of    global solutions without referring directly to
the fixed point equation (\ref{eq:5}) (since the fixed point condition
(\ref{eq:6}) now may fail). The  definition needs to be based directly
on (\ref{eq:equamotion}).  So we introduce (the equations written
slightly differently):
\begin{equation}\label{eq:equamotionB}
\begin{cases}
\big (\tfrac{\partial }{\partial t}u+P(Mu\cdot \nabla)u -\triangle
u\big )(t,\cdot)=0\mfor t\in I\\
u(t):=u(t,\cdot)\in\Ran P\mfor t\in \bar I
\end{cases}\;.
\end{equation} As in Subsection \ref{Abstract scheme} $I$ is an
interval of the form $]0,T]$ or of the form $I=]0,\infty[$. We shall  introduce a notion of
 strong solution to (\ref{eq:equamotionB}). The solutions  with
 $I=]0,\infty[$   will be called strong global solutions.  For that
purpose we need
the spaces appearing in Proposition \ref{prop:kato_small_data}. For
simplicity we shall restrict our discussion  to the $H^r$ setting
(leaving out the $\dot H^r$  setting with
$r\in[1/2,3/2[$).

So let $\vB=\vB_{I,3/8,5/4}$ and
$\vB^0\subseteq \vB$ be the  spaces as specified in the beginning of
Section \ref{Integral equation} (constructed in terms of an arbitrarily given
interval $I$).

\begin{defn}\label{defns:strong} Let  $r\geq 1/2$. For   $I=]0,T]$ we
  say that $u\in C(\bar I,H^r)$ is a {\it strong solution} to the
  problem
  (\ref{eq:equamotionB}) if    the following conditions hold:
  \begin{enumerate}[(1)]

  \item \label{item:20} $u(t)\in PH^r$ $\mforall t\in \bar I$,

  \item \label{item:21} $u\in \vB^0$,

\item \label{item:22} $u\in C^1(I,
  \vS'(\R^3) )$ and
\begin{equation}
    \label{eq:73}
   \tfrac{\d}{\d t}u=-A^2u-P(Mu\cdot \nabla)u;\,t\in I.
  \end{equation}
  \end{enumerate}

Here the differentiability in $t$ is meant in the weak* topology and in (\ref{eq:73}) is meant in the sense of distributions. The class of such functions is denoted by
  $\vS_{r,I}$.  For $I=]0,\infty[$ we define $\vG_{r}$ to be the
  subset of $C(\bar I,H^r)$ consisting of $u$'s  such that $1_{\tilde I}u\in
  \vS_{r,\tilde I}$ for all intervals of the form $\tilde I=]0,\tilde T]$, and we refer
  to any  $u\in \vG_{r}$ as  a {\it  strong global solution} to the
  problem
  (\ref{eq:equamotionB}) with  $I=]0,\infty[$.
\end{defn}

\begin{remarks}\label{rmks:strong-solutions}
  \begin{enumerate}[1)]

  \item \label{item:18} Obviously the condition \ref{item:21} is
    redundant if $r\geq 5/4$.
  \item \label{item:13}  For any  strong solution $u$  on $I$ the first  term on the right hand side of
    (\ref{eq:73}) is an element of  $C(I,H^{r-2})$
    while  the second
    term is  an element of $C( I,L^2)$, cf. (\ref{eq:8}). Consequently $u\in C^1(I,
  H^{\min (r-2,0)})$.
\item \label{item:14} For any  $u\in \vS_{r,I}$
  \begin{equation*}
   \tfrac{\d}{\d s} \parbb{
     \e^{-(t-s)A^2}u(s)}=-\e^{-(t-s)A^2}P(Mu(s)\cdot
\nabla)u(s)\mforall 0<s<t\in I,
  \end{equation*} and consequently (by integration) the integral
  equation (\ref{eq:1}) with $u_0=u(0)$ holds for all $t\in I$. In fact it follows
  that $X=u$ is a solution to (\ref{eq:5}) in $\vB^0$ (with
  $Y(t)=\e^{-tA^2}u_0$). Due to the uniqueness statement of Proposition \ref{prop:kato_small_data}
  it follows that $u$ coincides with the function $X$ of Proposition
  \ref{prop:kato_small_data} on a sufficiently small interval $\tilde I
  =]0,\tilde T]$. As a consequence similarly  if $r\in ]1/2, 3/2[$ or  $r\in
  [5/4,\infty[$, $u$ coincides  on a sufficiently small interval with the function $X$ of Propositions
  \ref{prop:kato small data2} or \ref{prop:kato small data2mod},
  respectively.
\item \label{item:15} Conversely, the solutions $X$ of
  Propositions
  \ref{prop:kato_small_data} and
   \ref{prop:kato big data} with data
  $u_0=Pu_0\in H^{1/2}$ are indeed solutions in the sense of Definitions
  \ref{defns:strong} (with $r=1/2$ and on  the same interval $I$).
  Similarly it is readily verified that the solution  $X$ of
  Propositions
  \ref{prop:kato small data2} or \ref{prop:kato small data2mod} with data
  $u_0=Pu_0\in H^{r}$ for  $r\in ]1/2, 3/2[$ or $r\in [5/4,\infty[$,
  respectively, is a  solution in the sense of Definition
  \ref{defns:strong} (with the same $r$ and $I$).

\item \label{item:16} The class $\vS_{r,I}$ is  right translation
  invariant, i.e. if $u\in \vS_{r,I}$  and $t_0\in ]0,T[$ (where $T$
  is the right end point of $I$)  then $u_{t_0}(\cdot):=
  u(\cdot+t_0)\in\vS_{r,I_0}$; $I_0:=]0,\infty[\cap\big (I-\{t_0\}\big
  )$. In particular $\vG_{r}$ is  right translation
  invariant ($u\in\vG_{r}\Rightarrow u_{t_0}\in\vG_{r}$ for any $t_0>0$).
\item \label{item:17} With the modification  of Definitions
  \ref{defns:strong} given by omitting \ref{item:20}, the previous
  discussion,
  \ref{item:18}--\ref{item:16},
  is  still appropriate (possibly  slightly
  modified).   Notice that we did not impose the
  condition \ref{item:20} (viewed as a condition on the data) in the bulk of the paper.
\end{enumerate}

\end{remarks}

Strong solutions to the same initial value problem are  unique:
\begin{prop}\label{prop:strong-solutions}
  Suppose $u_1\in \vS_{r_1,I_1}$ and $u_2\in \vS_{r_2,I_2}$ obey $u_1(0)=u_2(0)=u_0$ for some
  $u_0\in PH^{r_1}\cap PH^{r_2}$. Then $u_1=u_2$ on $I_1\cap I_2$.
\end{prop}
\begin{proof}
  This is a standard argument for ODE's.  We can assume that
  $r_1=r_2$ and $I:=I_1=I_2$.  Suppose $u_1\neq u_2$ on $I$. Then
  let
  \begin{equation*}
  t_0=\inf \{t\in I | u_1(t)\neq u_2(t)\}.
  \end{equation*}
 Clearly $t_0\in\bar I$, and by continuity
  $t_0<T$ and
  $u_1(t_0)= u_2(t_0)$. Due to Remark \ref{rmks:strong-solutions}
  \ref{item:16} we can assume  that $t_0=0$. Due to   Remark
  \ref{rmks:strong-solutions} \ref{item:14}
  it follows that $u_1(t)= u_2(t)$ for all
  $t\in \tilde I
  =]0,\tilde T]$  for some sufficiently small  $\tilde T>0$. This is a
  contradiction.

\end{proof}

\subsection{Sobolev and analyticity bounds for bounded
  intervals} \label{Sobolev and analiticity  bounds for finite
  intervals}
In this subsection we show that strong solutions are smooth, in fact
real analytic, in the $x$-variable.
\begin{prop}
  \label{prop:sobolev} Let $r\geq 1/2$ and $0<T_0<T<\infty$ be
  given. Let $u\in \vS_{r,I}$ where $I=]0,T]$, and denote by $|u|$ the
  norm $|u|=\|u\|_{\vB}$. There exist
  $\delta=\delta(T_0,|u|)>0$ and
  $C=C(T_0,T, |u|,\sup_{t\in I}\|u(t)\|_{H^{1/2}})>0$ such that
  \begin{equation}
    \label{eq:74}
    \|\e^{\delta A}u(t)\|_{H^{1/2}}\leq C \mforall t\in [T_0,T].
  \end{equation}
\end{prop}
  \begin{proof} For all $u\in
 \vB$ and $\widetilde T\in]0,T[$
    \begin{equation}
      \label{eq:75}
      \sup_{t\in [\tilde T,T]}\|A^{5/4}u(t)\|_{L^2}\leq \tilde T^{-3/8}|u|.
    \end{equation}

 For any given $u\in \vS_{r,I}$ we shall obtain an analyticity bound
 for the restriction of $u$ to $]t_0-\epsilon,t_0]$
 for
   $t_0\in  [T_0,T]$ and for suitable $\epsilon\in ]0,T_0[$. For that
we shall apply  the procedure of the proof of Theorem \ref{prop:ana_small_times1} to
the strong solution
$u_{t_0,\epsilon}:=u(\cdot +t_0-\epsilon)$  on
 the interval $I_{\epsilon}=]0,\epsilon]$. The application will
   be with $\lambda=1$ and $r=1/2$ in the definitions of $\zeta$ and $\theta$
   (given in (\ref{eq:20})) and for $\epsilon>0$ small, and the
   underlying Banach
   space will be $\tilde \vB=\vB_{\zeta,\theta,I_{\epsilon},3/8,5/4}$. Indeed for
   $\epsilon>0$ taken small enough the conditions
   (\ref{eq:6}) hold for some $R>0$ that can be taken independent of
   $T$ and $t_0 \in
   [T_0,T]$, cf. (\ref{eq:8bba}). Notice here that by Lemma \ref{Lemma:spect_bnd}  (with
   $\alpha=0$ and  $f=A^{5/4}u(t_0-\epsilon)$) and (\ref{eq:75})
   \begin{equation}
     \label{eq:76}
     |\e^{-(\cdot)A^2}u(t_0-\epsilon)|_{\tilde \vB}\leq
     \tilde C_1\epsilon^{3/8} \|A^{5/4}u(t_0-\epsilon)\|_{L^2}\leq
     \tilde C_1\epsilon^{3/8} (T_0-\epsilon)^{-3/8}|u|.
   \end{equation} So we can choose  $R$ in (\ref{eq:6}) to be equal to
   the constant on the right hand side of (\ref{eq:76}), and indeed the conditions
   (\ref{eq:6}) are  fulfilled for all
   sufficiently small  $\epsilon>0$, cf. (\ref{eq:8bba}). Fix any such
   $\epsilon>0$ and let $\delta=\sqrt \epsilon$. Then we invoke
   (\ref{eq:7cq}) and (\ref{eq:7bbbq}) with  $u=v=u_{t_0,\epsilon}$ and with time $t=\epsilon$ as
   well as Lemma \ref{Lemma:spect_bnd} (with $\alpha=0$, $f=\inp{A}^{1/2}u(t_0-\epsilon)$ and
   also applied for   $t=\epsilon$).

  \end{proof}

In combination with Theorems \ref{prop:ana_small_times1} and
\ref{prop:ana_small_times1mod} we obtain:
\begin{cor}
  \label{cor:sobol-analyt-bounds}
Let $r\geq 1/2$,  $0<T<\infty$ and $u\in \vS_{r,I}$ be
  given; $I:=]0,T]$.  Let $u_0=u(0)$. There exist
  $\delta=\delta(u,r)>0$ and
  $C=C(u,T,r)>0$ such that
  \begin{equation}
    \label{eq:74b}
    \|\e^{\min (\sqrt t, \delta) A}u(t)\|_{H^{1/2}}\leq C \mforall t\in I.
  \end{equation}
In particular, for all $\bar r\geq   1/2$ and with $\tilde C
=C\max_{x\geq 0}x^{\bar r-1/2}\e^{-x}$
\begin{equation}
    \label{eq:74bb}
    \|A^{\bar r}u(t)\|_{L^2}\leq \tilde C \min (\sqrt t, \delta)^{-(\bar r-1/2)}\mforall t\in I.
  \end{equation}
If $r>1/2$  the dependence  of $\delta$ and $C$ on
  $u$ can be chosen to be through $|u|$
  and $\|u_0\|_{H^r}$ and through $|u|$, $\sup_{t\in I}\|u(t)\|_{H^{1/2}}$
  and $\|u_0\|_{H^r}$, respectively.

For all $k\in
  \N\cup\{0\}$ and for all $\bar r\geq   1/2$
\begin{equation}
  \label{eq:80}
  u\in C^k(I,H^{\bar r}).
\end{equation}

Writing $u(t)(x)=u(t,x)$,
  \begin{equation}
    \label{eq:77}
    u\in C^{\infty}(I\times \R^3).
  \end{equation}

\end{cor}
\begin{proof} We apply Theorems \ref{prop:ana_small_times1} and
\ref{prop:ana_small_times1mod} with $\lambda=1$. There exist $T_0\in
]0,T[$ and $C>0$ such that
\begin{equation}
    \label{eq:74c}
    \|\e^{\sqrt t A}u(t)\|_{H^{1/2}}\leq C \mforall t\in ]0,T_0].
  \end{equation} These constants are for $r=1/2$ chosen in agreement
  with  an
  approximation property of $u_0$. This is not the case for $r>1/2$
  where the bounds (\ref{eq:18}) and (\ref{eq:18mod_anal}) can be used
  directly to get the  appropriate smallness in terms of the quantities
  $\|u_0\|_{H^r}$ and $r$. We shall use Proposition
  \ref{prop:sobolev} with the $T_0$ from (\ref{eq:74c}). Whence for $r>1/2$ we  apply Proposition
  \ref{prop:sobolev} with $T_0$ chosen as a function of the
  quantities $\|u_0\|_{H^r}$ and $r$. For $r=1/2$ we apply Proposition
  \ref{prop:sobolev} with   $T_0$ depending  on $u$ through $u_0$.

As for (\ref{eq:80}) with $k=0$ we  apply   (\ref{eq:74bb}) in combination
with Propositions \ref{prop:kato small data2mod} and
\ref{prop:strong-solutions} (notice that we can
assume that $\bar  r > 5/4$ from the very definition of
$\vS_{r,I}$). The statement (\ref{eq:80}) with arbitrary $k\geq 1$ follows
inductively by repeated differentiation of (\ref{eq:73}).

  The statement (\ref{eq:77}) follows from (\ref{eq:80}) and  the Sobolev embedding theorem.

\end{proof}
\begin{remark}
  \label{remark:sobol-analyt-bounds} For $r>1/2$ the Sobolev bounds (\ref{eq:74bb}) can
be improved  in the short time regime due to Theorems \ref{prop:ana_small_times1} and
\ref{prop:ana_small_times1mod}: For any given $\bar r\geq r$
the quantity has a  bound  of the form $Ct^{-(\bar r-r)/2}$ for
small $t>0$.
\end{remark}

The energy inequality was studied under certain conditions in
Subsection \ref{Energy inequality}. We can now prove it more generally:
\begin{cor}
  \label{cor:energybound2} Let $r\geq 1/2$, an interval $I=]0,T]$  and $u\in \vS_{r,I}$ be
  given.  Let $u_0=u(0)$. Suppose in addition  the condition
\begin{equation}
    \label{eq:54bb}
    \nabla\cdot \parb{Mu(t)}=0\mforall t\in I.
  \end{equation}
Then
\begin{equation}
    \label{eq:57bb}
    \|u(t)\|^2=\|u_0\|^2-2\int_0^t\|Au(s)\|^2\,\d s\mforall t\in I.
  \end{equation}

In particular $\|u(t)\|\leq \|u_0\|$  for all $ t\in I$.
\end{cor}
\begin{proof}
  Due to (\ref{eq:74bb}) (applied for the first identity with $\bar  r=3/2$ in combination
  with  Remark
  \ref{rmks:strong-solutions} \ref{item:13}) and (\ref{eq:54bb}) the computation
\begin{equation}
    \label{eq:56bb}
   \tfrac{\d }{\d t}\|u(t)\|^2=-2\|Au(t)\|^2
   +2\inp[\big]{u(t),(Mu(t)\cdot \nabla)u(t)}=-2\|Au(t)\|^2\mforall
   t\in I,
  \end{equation} is legitimate. By integration of (\ref{eq:56bb}) we
  obtain (\ref{eq:57bb}). Note incidentally  that
  $\|Au(s)\|^2=O(s^{-1/2})$, due to (\ref{eq:74bb}), yielding an
  independent proof of the convergence of  the
  integral in (\ref{eq:57bb}).

\end{proof}

\subsection{Global analyticity stability}\label{Completing the proof}
We shall study the set of data for which we have global
solutions. There are several works (for example \cite{PRST},\cite{GIP1}, \cite{GIP2}, \cite{ADT}, \cite{FO}, \cite{Zh}) which study the stability of solutions to the Navier-Stokes equations.  Perhaps the first result is \cite{PRST} but there are many further results for different spaces.  In particular the fact that $\vI_r$ defined below is an open set in our setting is a known result (\cite{GIP1}, \cite{GIP2}, \cite{ADT}).  Although we give the openness result we concentrate particularly on the stability of the region of analyticity and corresponding estimates.

We shall prove two stability results. The first is for  bounded
intervals only, however it is used in the proof of  our second (global)
stability result (and besides it has some independent interest, see for example Corollary \ref{cor:semicontinuity}):
\begin{prop}\label{thm:global-stability_0}  Let $I$ be an interval of
  the form $I=]0,T]$, and let $\theta: \bar I\to [0,\infty[$ be a continuous
  function obeying the following estimate for some $\lambda\geq 0$:
  \begin{equation}
    \label{eq:82}
\theta(s+t)\leq \lambda\sqrt s+\theta(t) \mfor s,t,s+t\in \bar I.
  \end{equation}
Suppose
  $u\in  \vS_{1/2,I}$ obeys
\begin{equation}
   \label{eqn:58new_0}
    A^{1/2}\e^{\theta(\cdot)
A}u(\cdot) \in C(\bar I, L^2).
\end{equation}
Let $u_0=u(0)$. There exists $\delta_0 >0$ such that:
\begin{enumerate}[i)]
  \item \label{item:12dd}
 If $\delta \leq  \delta_0, v_0
\in PH^{1/2}$  and $\|A^{1/2}\e^{\theta(0)A}(v_0 - u_0)\| \leq \delta$ it follows
that there exists
$v \in \vS_{1/2,I}$ with $v(0)=v_0$ obeying
\begin{subequations}
  \begin{align}
  \label{eqn:radiusstability_0}
  \|A^{1/2}\e^{\theta( t)A}(v(t)-u(t))\| &\leq K_1\delta,\\
  \label{eqn:radiusstability2_0}
  t^{3/8}\|A^{5/4}\e^{\theta( t)A}(v(t)-u(t))\|&\leq K_2\delta.
  \end{align}
%In particular  if $\mu < \lambda$ and $s > 0$
%\begin{equation}
%\label{eq:mustability_0}
%t^{s/2}\|A^{s + 1/2}\e^{\mu \sqrt tA}(v(t)-u(t))\|< CK\delta\text{
%  where }C=\sup_{x\geq 0}x^{s}\e^{(\mu -\lambda)x}.
%\end{equation}
\item \label{item:12dd2} If $\delta \leq  \delta_0, v_0
\in PH^{1/2}$ and $\|\e^{\theta(0)A}(v_0 - u_0)\|_{H^{1/2}} \leq \delta$ it follows
 in addition that
\begin{equation}
  \label{eqn:radiusstability3_0}
  \|\e^{\theta( t)A}(v(t)-u(t))\|\leq K_3\delta.
\end{equation}
\end{subequations}
\end{enumerate}
 In \eqref{eqn:radiusstability_0}--\eqref{eqn:radiusstability3_0}  the
 constants $K_1,K_2,K_3>0$ depend on $\theta$, $u$, $T$ and $\delta_0$ but not on
$\delta$, and all bounds are uniform in $t\in I$.
  \end{prop}

  In the proof we will use norms of the form
\begin{equation}
|w|_{s_0,t_0}: = \sup_{0 < s \leq
  \min (s_0,T-t_0)}s^{3/8}\|A^{5/4}\e^{\theta(s+t_0)A}w(s)\|;\;s_0>0,\,t_0\in [0,T[.
\end{equation}
Mimicking Subsection
    \ref{Analyticity bounds in
    dot r and  r
    for small times and for large data} thus with $\zeta(s)=1$ and
  $\theta(s)\to \theta(s+t_0)$  we find
\begin{subequations}
\begin{align}\label{eqn:A1halfbd0}
|B(w_1,w_2)|_{s_0,t_0} &\leq \gamma_{\lambda}|w_1|_{s_0,t_0}\cdot|w_2|_{s_0,t_0},\\
\label{eqn:A1halfbd}
\|A^{1/2}\e^{\theta(s+t_0)A}B(w_1,w_2)(s)\| &\leq
\gamma_{\lambda}|w_1|_{s_0,t_0}\cdot|w_2|_{s_0,t_0},\\
\label{eqn:A1halfbd2}
\|\e^{\theta(s+t_0)A}B(w_1,w_2)(s)\| &\leq
s^{1/4} \gamma_{\lambda}|w_1|_{s_0,t_0}\cdot|w_2|_{s_0,t_0}.
\end{align}
\end{subequations}
Here $\gamma_{\lambda} = dc_{\lambda}$ where  $d$ is independent of $\lambda,
s_0$, $t_0$ and $T$ and
$$c_{\lambda} := \sup_{x\geq0}\,\inp{x}^{5/4}\e^{\lambda x}\e^{-x^2}.$$
We will need the following lemma:

\begin{lemma}\label{lemma:uniformboundonu}
Suppose $\theta$ and  $u \in \vS_{1/2,I}$  are given as in
Proposition \ref{thm:global-stability_0}. Let  $\epsilon \in \ ]0,(2\gamma_{\lambda})^{-1}[$
be given. Then there is an $s_0 \in ]0,1[, s_0 =
s_0(\epsilon,\theta,u)$,  so that
\begin{align}&\forall t_0 \in [0,T[ \;\forall s\in ]0,\min
  (s_0,T-t_0)]: \label{eq:unifbdu}\\
    &2s^{3/8}\|A^{5/4}\e^{\theta(s+t_0)A}\e^{-sA^2}u(t_0)\| \leq
    \epsilon \mand s^{3/8}\|A^{5/4}\e^{\theta(s+t_0)A}u(s+t_0)\| \leq
    \epsilon.\nonumber
 \end{align}
\end{lemma}
\begin{proof}
Using $\theta(s+t_0) \leq \lambda\sqrt{s} + \theta(t_0)$  and the spectral theorem we have for any $N > 0$
\begin{eqnarray}\label{eq:smallness}
    && s^{3/8}\|A^{5/4}\e^{\theta(s+t_0)A}\e^{-sA^2}u(t_0)\|\nonumber\\
&\leq &\|(\sqrt{s}A)^{3/4}\e^{\lambda\sqrt{s}A}\e^{-sA^2}\|\cdot\|1_{[N,\infty[}(A)A^{1/2}\e^{\theta(t_0)A}u(t_0)\nonumber\| \\
&+& s^{3/8}\|\e^{\lambda\sqrt{s}A}\e^{-sA^2}\|\cdot\|A^{3/4}1_{[0,N]}(A)A^{1/2}\e^{\theta(t_0)A}u(t_0)\|\nonumber\\
&\leq & c_{\lambda}\|1_{[N,\infty[}(A)A^{1/2}\e^{\theta(t_0)A}u(t_0)\| +c_{\lambda}s^{3/8}N^{3/4}\|A^{1/2}\e^{\theta(t_0)A}u(t_0)\|.
\end{eqnarray}
Since the map ${\bar I\ni t_0 \to A^{1/2}\e^{\theta(t_0)A}u(t_0)}$ is
continuous, it maps into a compact set on which ${1_{[N,\infty[}(A)
  \to 0}$ uniformly as ${N \to \infty}$. We then fix $N$ so that the
first term of (\ref{eq:smallness}) is less than $\epsilon/4$ for all
$t_0 \in \bar I$.  Once $N$ is fixed we can choose $s_0 \in ]0,1[$ so
that the second term in (\ref{eq:smallness}) is less than $\epsilon/4$
for $s \in [0,s_0]$. We have proved the first estimate of
(\ref{eq:unifbdu}), $|\e^{-(\cdot)A^2}u(t_0)|_{s_0,t_0}\leq \epsilon/2$.

 To show the second estimate of
(\ref{eq:unifbdu}) we go back to the integral equation \eqref{eq:1}
and use $u(t_0)$ as initial data following the scheme of Subsection
\ref{Abstract scheme} (with $R=\epsilon/2$). We use the first
estimate in combination with (\ref{eqn:A1halfbd0}). By uniqueness the
constructed fixed point $w=u_{t_0}$ where $u_{t_0}(s)=u(s+t_0)$.
\end{proof}

\begin{proof}[Proof of Proposition \ref{thm:global-stability_0}]
We now choose  $\epsilon = (3\gamma_{\lambda})^{-1}$ and $s_0$ in
accordance with Lemma \ref{lemma:uniformboundonu}. We can assume that
$m_0:=  T/s_0\in \N$.  We build the solution $v$ in the interval $I$ by constructing it in a series of intervals $[(m-1)s_0,ms_0], m = 1,2, \dots, m_0$. We assume inductively we have constructed $v(t)$ in the interval $0\leq t \leq ms_0$ (with $m \leq m_0 -1$) and that we have the estimate
\begin{equation}
   \label{eq:inductivebd}
   \|A^{1/2}\e^{\theta(t)A}(v(t) - u(t))\| \leq (2c_{\lambda})^m\delta
\end{equation}
in this interval (this is true for $m = 0$). Let $t_0 = ms_0$ and $u_{t_0}(s) = u(s+ t_0)$. Consider the map
\begin{equation}\label{eqn:Fdef}
F(w)(s): = \e^{-sA^2}(v(t_0) - u(t_0)) + B(w,u_{t_0})(s) + B(u_{t_0},w)(s) + B(w,w)(s).
\end{equation}
 We have
\begin{subequations}
\begin{align}\label{eq:81}
|F(w)|_{s_0,t_0} &\leq c_{\lambda}(2c_{\lambda})^m\delta + 2\gamma_{\lambda}|w|_{s_0,t_0}\cdot|u_{t_0}|_{s_0,t_0} + \gamma_{\lambda}|w|_{s_0,t_0}^2,\\
|F(w_1) - F(w_2)|_{s_0,t_0} &\leq \gamma_{\lambda}(2|u_{t_0}|_{s_0,t_0} + |w_1|_{s_0,t_0} + |w_2|_{s_0,t_0})\cdot|w_1 - w_2|_{s_0,t_0}.\label{eq:83}
\end{align}
\end{subequations}
Then a simple computation shows $F:B_{2R} \to B_{2R}$ is a strict contraction
if $R = c_{\lambda}(2c_{\lambda})^m\delta$ and $\delta\leq \delta_0$
where $\delta_0>0$ is chosen small enough.  If the fixed point is
denoted by $w$, we define $v(t) = u(t) + w(t-t_0)$ for $t \in
[ms_0,(m+1)s_0]$.  The bound (\ref{eq:inductivebd}) with $m
\rightarrow m+1$ in the interval $[ms_0,(m+1)s_0]$ follows from $w =
F(w)$ and the estimate (\ref{eqn:A1halfbd}).  This completes the
induction and gives (\ref{eq:inductivebd}) with $m = m_0$ for $t \in
I$. We have constructed a solution $v$ obeying
(\ref{eqn:radiusstability_0}). From the very construction we have partly
shown (\ref{eqn:radiusstability2_0}), however our bounds are somewhat
poor at $ms_0$, $m = 1, \dots, m_0-1$ (assuming here $m_0\geq 2$). In order to show
(\ref{eqn:radiusstability2}) near $ms_0$, $m = 1, \dots, m_0-1$,  we can repeat the above
procedure in the intervals $[(m-1/2)s_0, (m+1/2)s_0]$. The consistency
of our definitions in overlapping intervals follows from
uniqueness. For (\ref{eqn:radiusstability3_0}) we use
(\ref{eqn:A1halfbd2}) to show inductively
\begin{equation}
   \label{eq:inductivebdss}
   \|\e^{\theta(t)A}(v(t) - u(t))\| \leq (2c_{\lambda})^m\delta \mfor  0\leq t \leq ms_0,
\end{equation} cf. (\ref{eq:inductivebd}).
\end{proof}

We will use Proposition \ref{thm:global-stability_0} to shed some light on (\ref{eq:semi-continuity}), the conjectured lower semicontinuity of the analyticity radius of $u \in \vS_{1/2,I}$.  To motivate the construction in the following corollary, it should be noted that by definition of $\rad(u(t))$, if $t > 0$ and
\begin{equation}\label{eqn:lsc in time}
\liminf_{s \uparrow t} \rad(u(s)) \geq \rad(u(t))
\end{equation}
then
\begin{equation}
\forall \, \alpha < \rad(u(t)) \, \exists \, a < t \, \text{so that} \,  \|\e^{\alpha A}u(s)\| < \infty \, \forall s  \in ]a,t]
\end{equation}
but the uniform bound
\begin{equation}\label{eqn:unifbd}
\forall \, \alpha < \rad(u(t)) \, \exists \, a < t \, \text{so that} \,  \sup_{s\in]a,t]}\|\e^{\alpha A}u(s)\| < \infty
\end{equation}
does not readily follow from the definitions.

\begin{cor}\label{cor:semicontinuity}
Fix $u \in \vS_{1/2,I}$ and let $u_0 = u(0)$.  For $t \in I$ and any $v \in \vS_{1/2,I}$ let $v_0 = v(0)$ and define
\begin{equation} \label{eq:newrad}
r_t(v_0): = \sup\left\{\alpha \geq 0| \sup_{s\in ]a,t]} \|\e^{\alpha A}v(s)\| < \infty\; {\rm for \; some} \, a < t\right\}
\end{equation}
Then $r_t(\cdot)$ is lower semicontinous at $u_0$ as a function of the initial data $v_0$ in the $H^{1/2}$ topology.  More precisely, for $t \in I$
\begin{equation}\label{eqn:lscnewrad}
\liminf_{\,\,\|v_0 - u_0\|_{H^{1/2}} \rightarrow  0} r_t(v_0) \geq r_t(u_0).
\end{equation}

If the analyticity radius of $u$ satisfies (\ref{eqn:lsc in time}) and (\ref{eqn:unifbd}) then $r_t(u_0) = \rad(u(t))$ and the analyticity radius at $t$ is lower semicontinous as a function of the initial data at $u_0$.  More precisely, (\ref{eq:semi-continuity}) is valid.
\end{cor}

\begin{proof}
Fix $t \in I$.  Without loss we can assume $r_t(u_0) > 0$ and $ I = ]0,t]$. Choose $0 < \alpha < r_t(u_0)$.  Then there exists $0 < a < t$ so that $\sup_{s\in ]a,t]} \|\e^{\alpha A}u(s)\| < \infty$.  Define $\theta : [0,t] \rightarrow [0,\infty[$:
\begin{equation}
\theta(s) =
\begin{cases} 0, &\text{ if $s\in [0,a]$};\\
             (\frac {s-a}{t-a}) \alpha, &\text{if $s\in [a,t]$.}
\end{cases}
\end{equation}
Note that $\theta(\tau + s) \leq \lambda\sqrt \tau + \theta(s)$ with $\lambda = \alpha (t-a)^{-1/2}$ so that Proposition \ref{thm:global-stability_0} applies. It follows that if $\|v_0 - u_0\|_{H^{1/2}}$ is small enough, $\sup_{\tau \in ]a,t]} \|\e^{\theta(\tau) A}v(\tau)\| < \infty$.  Thus by definition, for these $v_0, r_t(v_0) \geq \alpha$. This gives (\ref{eqn:lscnewrad}).
As for the last statement of the corollary, following through the definitions it is easy to see that $r_t(u_0) = \rad(u(t))$ under the stated conditions. The definition of $r_t(v_0)$ also implies $ \rad(v(t)) \geq r_t(v_0)$ for any $v \in \vS_{1/2,I}$ and thus (\ref{eqn:lscnewrad}) gives the stated result.
\end{proof}

We now continue with our discussion of global stability.

\begin{defn}
  \label{def:global-stability}For   $r\geq 1/2$ we denote by
\begin{equation}
  \label{eq:78}
  \vI_r=\{u_0\in PH^r| \,\exists u\in \vG_r: \,u(0)=u_0\},
\end{equation} and we endow $ \vI_r$ with the topology from the space $PH^r$.
\end{defn}

Our result on global stability is as follows:
\begin{thm}\label{thm:global-stability}  Suppose
  $u_0\in  \vI_{1/2}$ and that the corresponding strong global solution
  $u$ obeys
\begin{equation}
    \label{eq:57bbb}
    \liminf_{t\to \infty }\|A^{1/2}u(t)\|=0.
  \end{equation}
  Suppose in addition that $\lambda > 0$ is given so that
\begin{equation}
   \label{eqn:58new}
    A^{1/2}\e^{\lambda \sqrt \cdot
A}u(\cdot) \in C([0,\infty[, L^2).
\end{equation}
There exists $\delta_0>0$ such that:
\begin{enumerate}[i)]
  \item \label{item:s1}
If $\delta \leq \delta_0, v_0
\in PH^{1/2}$ and $\|A^{1/2}(v_0 - u_0)\| \leq \delta$ it follows that
$v_0 \in \vI_{1/2}$ and that for all $t > 0$ the corresponding strong
global solution $v$ satisfies
\begin{subequations}
  \begin{align}
  \label{eqn:radiusstability}
  \|A^{1/2}\e^{\lambda \sqrt tA}(v(t)-u(t))\| &\leq K_1\delta,\\
  \label{eqn:radiusstability2}
  t^{3/8}\|A^{5/4}\e^{\lambda \sqrt tA}(v(t)-u(t))\|&\leq  K_2\delta.
  \end{align}
%Here $K>0$ depends on $u_0$, $\lambda$, and $\delta_0$ but not on
%$\delta$.
%In particular  if $\mu < \lambda$ and $s > 0$
%\begin{equation}
%\label{eq:mustability}
%t^{s/2}\|A^{s + 1/2}\e^{\mu \sqrt tA}(v(t)-u(t))\|< CK\delta\text{
%  where }C=\sup_{x\geq 0}x^{s}\e^{(\mu -\lambda)x}.
%\end{equation}
\item \label{item:s2} If in addition $u_0 \in PH^r$ with $r > 1/2$, then $u_0\in
\vI_r$ and $u_0$ is an interior point of $\vI_r$.
\item \label{item:s3} If $\delta \leq  \delta_0, v_0
\in PH^{1/2}$ and $\|v_0 - u_0\|_{H^{1/2}} \leq \delta$ it follows
 in addition that
\begin{equation}
  \label{eqn:radiusstability32}
  \inp{t}^{-1/4}\|\e^{\lambda \sqrt tA}(v(t)-u(t))\|\leq K_3\delta.
\end{equation}
\end{subequations}
\end{enumerate}
In \eqref{eqn:radiusstability}--\eqref{eqn:radiusstability32}  the
 constants $K_1,K_2,K_3>0$ depend on $\lambda$, $u$, and $\delta_0$ but not on
$\delta$, and all bounds are uniform in $t>0$.
  \end{thm}
\begin{proof} Choose $T_0 > 1$ and large enough so that with $w(s) = u_{T_0}(s)=u(s+T_0)$
\begin{equation}\label{eqn:T0toinfinity}
|w|_{\infty,T_0}: = \sup_{s> 0}s^{3/8}\|A^{5/4}\e^{\lambda\sqrt{s+T_0}A}w(s)\| \leq (3\gamma_{\lambda})^{-1}.
\end{equation}
This is possible by Corollary \ref{cor:o(1)}.    On the one hand we apply
Proposition \ref{thm:global-stability_0} with $T=2T_0$ and
$\theta(t)=\lambda\sqrt t$ to construct a solution $v$ in the interval
$[0,T]$. We now construct $v$ in the interval $[T_0,\infty[$ using the
bound (\ref{eqn:T0toinfinity}).  This is done in a similar way as in
the proof of Proposition \ref{thm:global-stability_0} using the map $F$ defined in (\ref{eqn:Fdef}) with the replacement $t_0 \rightarrow T_0$ and using the Banach space with norm $|\cdot|_{\infty,T_0}$ defined in (\ref{eqn:T0toinfinity}).  The contraction mapping argument then gives a fixed point $w$ with
$|w|_{\infty,T_0} \leq K \delta$.  Finally we extend $v$ to
$[T_0, \infty[$ by setting $v(t) = u(t) + w(t-T_0)$.   The
estimates (\ref{eqn:radiusstability})--(\ref{eqn:radiusstability32})
follow easily. The statement \ref{item:s2} follows from \ref{item:s1}
in combination with Propositions \ref{prop:kato small data2} and
\ref{prop:kato small data2mod}  and Corollary \ref{cor:sobol-analyt-bounds}.

\begin{comment}
To see that if $u_0 \in \vI_r$ then $u_0$ is an interior point it
suffices to show that if in addition $v_0 \in PH^r$ then the solution
$v$ constructed above is in $\vG_r$.  From (\ref{eq:mustability}) and
the contraction mapping argument, it follows that $v-u \in C([t_0,T]),
PH^s)$ for any $s$ if $0 < t_0 < T < \infty$.  By Propositions
\ref{prop:kato small data2} and \ref{prop:kato small data2mod}, along
with the second remark in Remarks \ref{remarks:uniq}, $v-u \in
C([0,t_0], PH^r)$ for small $t_0$ and it satisfies (2) of Definition
\ref{defns:strong}.
\end{comment}
\end{proof}

\begin{lemma}\label{lemma:energyre} Let  $r\geq 1/2$. Suppose a given $u\in \vG _r$
  obeys (\ref{eq:54bb}) with $I=]0,\infty[$. Then  the condition
  (\ref{eq:57bbb}) holds.
  \begin{proof}
   Let
  $u_0=u(0)$, and let $n\in \N$ be given. Pick $\delta>0$ such that
    $\delta\|u_0\|< n^{-2}$, and pick $\tilde n\geq n$ such that
    $\delta^22\tilde n> \|u_0\|^2$. Then, due to  (\ref{eq:57bb}) with $t=2\tilde n$, for
    some $t_n\in ]\tilde n,2\tilde n]$ we have $\|Au(t_n)\|<\delta$. For this time $t_n$
    \begin{equation*}
      \|A^{1/2}u(t_n)\|^2\leq \|Au(t_n)\|\,\|u(t_n)\|\leq \delta \|u_0\| < n^{-2},
    \end{equation*} so $\|A^{1/2}u(t_n)\|<1/n$. Whence $\lim_{n\to
      \infty}\|A^{1/2}u(t_n)\|=0$ for some sequence $t_n\to \infty$.
  \end{proof}
 \end{lemma}
Using Theorem  \ref{thm:global-stability} and  Lemma  \ref{lemma:energyre} we obtain for the system
(\ref{eq:equamotiona}):
\begin{cor}
  \label{cor:global-stability} Let  $M=I$ and $P$ be
  given as  the Leray projection. Then for all   $r\geq 1/2$ the set  $\vI_r$ is open in   $PH^r$.
\end{cor}
\subsection{$L^2$ stability}\label{L^2 decay stability}

Our final result on the stability of the $L^2$ norm is motivated by
various previous works on  $L^2$ decay properties, in particular
\cite{Sc1,Sc2,Wi}. In the following main result note the asymmetry between the solutions $u$ and $v$ reflected in the dependence of the constant $K$ in (\ref{eqn:L2stability}) on $u$.

\begin{prop}\label{prop:L2stability} Suppose $u, v \in \vG_{1/2}$ where $M = I$ and $P$ is the Leray projection.  Let $u_0(t) = \e^{-tA^2}u(0)$ and $v_0(t) = \e^{-tA^2}v(0)$.  We suppose
\begin{equation} \label{eqn:u0v0bound}
\|u_0(t)\| + \|v_0(t)\| \leq L \langle t \rangle^{-\sigma/2}
\end{equation}
where $\sigma \geq 0$.
Let $z(t) = v(t) - u(t) -w_0(t)$ where $w_0(t) = v_0(t) - u_0(t)$.

There exists $\delta_0 >0$ such that if $0 < \delta \leq  \delta_0$ and
$\|v(0) - u(0)\|_{H^{1/2}} \leq \delta$ we have for any $\epsilon \in [0,1]$ such that $(1-\epsilon/2)\sigma \neq 1$,
\begin{equation}\label{eqn:L2stability}
\|z(t)\| \leq K\delta^{\epsilon} \langle t \rangle^{-\rm min((1-\epsilon/2)\sigma+1/4,\, 5/4)}.
\end{equation}
Here $K$ depends on $L$, $\delta_0$, $\epsilon$, and $u$.
\end{prop}
\begin{remarks}\label{rmks:globa}
\begin{enumerate}[1)]
  \item \label{item:19p}
The condition $\|u_0(t)\| = O(t^{-\sigma/2})$ of \eqref{eqn:u0v0bound}
is equivalent
%(by an integration by parts using $d/dr \|1_{[0,r]}(A)f\|^2 = 4\pi
%r^2\int_{S^2}|f(r\omega)|^2d\omega$)%
 to the condition
$\|1_{[0,r]}(A)u_0(0)\| = O(r^{\sigma})$.  This is one route to familiar sufficient conditions in terms of the
$L^p$ norms of $u_0(0)$ and of $xu_0(0)$.  For these conditions and
additional inequalities see \cite{BJ}.
\item \label{item:20p} The condition $\|u_0(t)\| = O(t^{-\sigma/2})$
  of \eqref{eqn:u0v0bound} implies the following decay of the solution
  $u\in \vG_{1/2}$,
\begin{equation} \label{eqn:uL2bound}
\|u(t)\| \leq C\langle t \rangle^{-\rm min(\sigma/2,\ 5/4)}.
\end{equation} The inequality (\ref{eqn:uL2bound}) is proved in
\cite{Wi}. It follows from an argument similar to but simpler than an
argument used in the proof of Proposition \ref{prop:L2stability}  to follow. Thus we omit the proof.
\item \label{item:21p} The positive parameter $\delta_0$ of Proposition
  \ref{prop:L2stability} can be determined as follows: Choose
  $u_0=u(0)$ in Theorem
  \ref{thm:global-stability} and $\lambda >0$ in agreement with
  \eqref{eqn:58new}. Then according to Theorem
  \ref{thm:global-stability} there exists  $\delta_0>0$ so that
  \eqref{eqn:radiusstability32} holds. This  $\delta_0$ applies
  in  Proposition
  \ref{prop:L2stability} (in fact we shall only need
  \eqref{eqn:radiusstability32}  with $\lambda=0$).
\item \label{item:22p}
 The condition $(1-\epsilon/2)\sigma \neq 1$ is introduced for
 simplicity to avoid logarithms in (\ref{eqn:L2stability}). If in
 addition to the hypotheses of  Proposition
  \ref{prop:L2stability} (excluding the requirement on $\epsilon$) we
  require $\|w_0(t)\| \leq \delta \langle t
  \rangle^{-\sigma/2}$, then by repeating the  proof of Proposition
  \ref{prop:L2stability} the estimate (\ref{eqn:L2stability}) can be improved to
\begin{align}
\|z(t)\| \leq
\begin{cases}
     K\delta\langle t \rangle^{-\rm min(\sigma+1/4,\, 5/4)},    &\text{if $\sigma \neq 1$}\\
     K\delta\langle t \rangle^{-5/4}\ln(t+2),                    &\text{if $\sigma = 1$}
\end{cases}\; .
\end{align}
\end{enumerate}
\end{remarks}
We will need the following lemma:
\begin{lemma}\label{lemma:gradbounds}
Assume the hypotheses of Proposition \ref{prop:L2stability} and in
addition the bound \ $\|u(t)\| \leq L \langle t \rangle^{-\sigma/2}$.
Then for any $\epsilon \in [0,1]$, and $0 < \delta \leq 1$,
\begin{subequations}
  \begin{align}
  \label{eqn:gradest}
\|\nabla w_0(t)\|_{\infty} &\leq C\delta^{\epsilon} t^{-1}\langle t \rangle^{-(1/4 + (1-\epsilon) \sigma/2)},\\
\|\nabla u(t)\|_{\infty} &\leq C t^{-1}\langle t \rangle^{-(1/4 + \sigma/2)},\label{eqn:gradestbb}\\
\label{eqn:L2interp}
\|w_0(t)\| &\leq C\delta^{\epsilon}\langle t \rangle^{-(1-\epsilon)
  \sigma/2}.
\end{align}
\end{subequations}
\end{lemma}

\begin{proof}
In the following the definition of $C$ may change from line to line. We have
\begin{equation}
\|\nabla w_0(t)\|_{\infty} = \|\nabla \e^{-tA^2}w_0(0)\|_{\infty} = \|K_0(t)*w_0(0)\|_{\infty} \leq \delta \|K_0(t)\|
\end{equation}
where $K_0(t)$ is a (constant times) the inverse Fourier transform of $\xi \e^{-t|\xi|^2}$.  We easily calculate  $\|K_0(t)\| = Ct^{-5/4}$ and thus $\|\nabla w_0(t)\|_{\infty} \leq C\delta t^{-5/4}$. We also have
$$\|\nabla w_0(t)\|_{\infty} = \|\nabla \e^{-tA^2/2}w_0(t/2)\|_{\infty} = \|K_0(t/2)*w_0(t/2)\|_{\infty} \leq 2L C t^{-5/4}\cdot t^{-\sigma/2}.$$
Thus interpolating these two results we find for large time
\begin{equation}
\|\nabla w_0(t)\|_{\infty} \leq C\delta^{\epsilon}t^{-5/4 - (1-\epsilon)\sigma/2}.
\end{equation}
For small and intermediate times
\begin{equation} \label{eqn:gradw0smallt}
\|\nabla w_0(t)\|_{\infty} = \|\nabla A^{-1/2}\e^{-tA^2} A^{1/2}w_0(0)\|_{\infty} = \|K_1(t)*A^{1/2}w_0(0)\|_{\infty} \leq \delta \|K_1(t)\|
\end{equation}
where $K_1(t)$ is a (constant times) the inverse Fourier transform of
$\xi|\xi|^{-1/2}\e^{-t|\xi|^2}$.  We easily calculate $\|K_1(t)\| =
Ct^{-1}$ and thus we have proved \eqref{eqn:gradest}.

The proof of \eqref{eqn:gradestbb} goes along the same lines. We
obtain (taking here $\lambda>0$ sufficiently small)
\begin{equation}
\|\nabla u(t)\|_{\infty} = \|\nabla \e^{-\lambda \sqrt tA}\e^{\lambda \sqrt tA}u(t)\|_{\infty} = \|K_2(t)*\e^{\lambda \sqrt tA}u(t)\|_{\infty} \leq \|K_2(t)\|\cdot G_0(t)^{1/2}
\end{equation}
where here  $K_2(t)$ is a (constant times) the inverse Fourier transform of $\xi \e^{-\lambda \sqrt t|\xi|}$ and $G_0(t) = \|\e^{\lambda \sqrt t A} u(t)\|^2$.  We calculate $\|K_2(t)\| = Ct^{-5/4}$ and using Corollary \ref{cor:decay1} \ref{item:11} we deduce $G_0(t) \leq C t^{-\sigma}$.  Thus we have the large time estimate
\begin{equation}
\|\nabla u(t)\|_{\infty} \leq Ct^{-5/4 - \sigma/2}.
\end{equation}
To obtain the result for small and intermediate times we use (\ref{eqn:58new}) which gives
\begin{align}
\|\nabla u(t)\|_{\infty} & = \|\nabla A^{-1/2}\e^{-\lambda \sqrt tA} A^{1/2}\e^{\lambda \sqrt tA}u(t)\|_{\infty} = \|K_3(t)*A^{1/2}\e^{\lambda \sqrt tA}u(t)\|_{\infty} \notag \\
 & \leq C\|K_3(t)\|.
\end{align}
We calculate $\|K_3(t)\| = Ct^{-1}$, completing the proof of
(\ref{eqn:gradestbb}).

The proof of (\ref{eqn:L2interp}) proceeds by interpolating the two bounds $\|w_0(t)\| \leq \delta$ and $\|w_0(t)\| \leq 2L\langle t \rangle^{-\sigma/2}$.
\end{proof}
\begin{remark*}
We will not use the bounds in (\ref{eqn:gradest}) and
(\ref{eqn:gradestbb}) for small $t$.  We include these estimates for completeness.
\end{remark*}

\begin{proof}[Proof of Proposition \ref{prop:L2stability}]
\begin{comment}}
We note
\begin{equation}
z(t) = - P\int_0^te^{-(t-s)A^2}\big(w(s)\cdot\nabla u(s) + u(s)\cdot\nabla w(s) + w(s)\cdot\nabla w(s)\big)ds
\end{equation}
where $w = v - u$. We use the norm $|X| = \sup_{t > 0}\|t^{3/8}A^{5/4}X(t)\|$ to obtain
\begin{equation}
\|z(t)\| \leq C (2|w|\cdot|u| + |w|^2)t^{1/4}.
\end{equation}
Using the fact that $u \in \vG_{1/2}$ and Proposition \ref{prop:kato
  big data} we conclude that $|u| < \infty$.  From
(\ref{eqn:radiusstability2}) we have $|w| \leq K\delta$ (cf. Remark
\ref{rmks:globa} \ref{item:21p}).  We thus obtain
\begin{equation} \label{eqn:zeesmalltime}
\|z(t)\| \leq C \delta t^{1/4}, t \geq 0.
\end{equation}
\end{comment}
Note that from (\ref{eqn:radiusstability32}) we have
\begin{equation} \label{eqn:zeesmalltime}
\|z(t)\| \leq C \delta \inp{t}^{1/4}, t \geq 0.
\end{equation}

For large time we use Schonbek's technique \cite{Sc1}, \cite{Sc2}.
We differentiate $\|z(t)\|^2$ and find (after an integration by parts)
\begin{equation}\label{eqn:zeesq}
\frac{1}{2}\frac{d}{dt}\|z(t)\|^2 = -\|Az\|^2 -\langle z,z\cdot\nabla (u + w_0)\rangle - \langle z,u\cdot\nabla w_0 + w_0\cdot\nabla u \rangle - \langle z,w_0\cdot\nabla w_0 \rangle.
\end{equation}

 We will need the estimate \eqref{eqn:uL2bound}.   Without loss of
 generality we can assume $\sigma \leq 5/2$. In the following we
 shall also assume  that $\delta_0 \leq 1$ (in addition to the
 requirement of Remark
\ref{rmks:globa} \ref{item:21p}); this is just to minimize notation.

We have (with some abbreviation)
\begin{equation}\label{eqn:zeehat}
\hat{z}(t,\xi) = -P(\xi)\int_0^t \e^{-(t-s)|\xi|^2}i\xi[\widehat{u\otimes w}(s,\xi)+ \widehat{w\otimes u}(s,\xi) + \widehat{w\otimes w}(s,\xi)]ds
\end{equation}
where $P(\xi)$ is the Leray projection in Fourier space. As a first step we get an initial bound for $\|z(t)\|$ of the form $\|z(t)\| \leq C \delta^{\epsilon}$ for all $t>0$. Define $k$ so that sup$_{t\geq 0}\|z(t)\| = k\delta^{\epsilon}$. In the last term of (\ref{eqn:zeehat}) we estimate
$$|\widehat{w\otimes w}(s,\xi)| \leq \|w(s)\|^2 \leq (\|u(s)\| + \|v(s)\|)\cdot \|w(s)\| \leq (2\|u(0)\| + \delta)\|w(s)\|.$$
We note that by definition of $k$
$$\|w(s)\| \leq \|w(0)\| + k\delta^{\epsilon} \leq \delta + k\delta^{\epsilon} \leq (1 + k)\delta^{\epsilon}.$$
From (\ref{eqn:zeehat}) it thus follows that
\begin{equation} \label{eqn:zhatbound1}
|\hat{z}(t,\xi)| \leq C|\xi|^{-1}(1+k)\delta^{\epsilon}.
\end{equation}
Using (\ref{eqn:zeesq}) we obtain for $t\geq 1$
\begin{align}\label{eqn:zdiffineq}
\frac{1}{2}\frac{d}{dt}\|z(t)\|^2 & \leq -\int_{|\xi|^2t \geq a}|\xi|^2|\hat{z}(t,\xi)|^2 \,d\xi + \|z(t)\|^2(\|\nabla w_0(t)\|_{\infty} + \|\nabla u(t)\|_{\infty})  \notag \\
& + \|z(t)\|\big(\|u(t)\|\cdot \|\nabla w_0(t)\|_{\infty} + \|w_0(t)\|\cdot \|\nabla u(t)\|_{\infty} + \|w_0(t)\|\cdot \|\nabla w_0(t)\|_{\infty}\big) \notag \\
& \leq -(at^{-1}- Ct^{-5/4 - \sigma/2})\|z(t)\|^2 + \|z(t)\|C\delta^{\epsilon}t^{-5/4 - (1-\epsilon/2)\sigma} + at^{-1}\int_{|\xi|^2t \leq a}|\hat{z}(t,\xi)|^2 \,d\xi \notag \\
& \leq -(a/2t)\|z(t)\|^2 + C\delta^{2\epsilon}t^{-3/2 -2(1-\epsilon/2)\sigma} + at^{-1}\int_{|\xi|^2t \leq a}|\hat{z}(t,\xi)|^2 \,d\xi.
\end{align}
Here we have used \eqref{eqn:uL2bound}, (\ref{eqn:gradest})--(\ref{eqn:L2interp}), and the "squaring inequality" for the term linear in $\|z(t)\|$.  We have also taken $a$ large. Inserting (\ref{eqn:zhatbound1}) and integrating we find for $a$ large enough and $t \geq 1$
$$\|z(t)\| \leq  t^{-a/2}k\delta^{\epsilon} + Ct^{-1/4}(1+k)\delta^{\epsilon}.$$
We take $t_0  > 1$ large enough so that for $t \geq t_0$
$$\|z(t)\| \leq k\delta^{\epsilon}/3 + (1+k)\delta^{\epsilon}/3.$$
If the supremum of $\|z(t)\|$ does not occur for $t \leq t_0$ then
$$k\delta^{\epsilon} \leq k\delta^{\epsilon}/3 + (1+ k)\delta^{\epsilon}/3$$
so $k \leq 1$.  Otherwise $k\delta^{\epsilon} \leq C \delta
\inp{t_0}^{1/4}$ by (\ref{eqn:zeesmalltime}) so that  $ k \leq C \inp{t_0}^{1/4}$.  Thus $\|z(t)\| \leq C\delta^{\epsilon}$.
We now make the inductive assumption that $\|z(t)\| \leq C\delta^{\epsilon}\langle t \rangle^{-\mu/2}$.  From (\ref{eqn:zeehat}) we obtain
\begin{align}
|\hat{z}(t,\xi)| & \leq |\xi|\int_0^t(2\|u(s)\|\cdot\|w(s)\| + \|w(s)\|^2) \ ds \notag \\
& \leq C|\xi|\int_0^t\big(\|u(s)\|\cdot(\|w_0(s)\|+ \|z(s)\|) + \|w_0(s)\|^2 + \|z(s)\|^2\big) \ ds \notag \\
& \leq C|\xi|\int_0^t\delta^{\epsilon}\big(\langle s \rangle^{-(1-\epsilon/2  )\sigma} +\langle s \rangle^{-(\sigma + \mu  )/2} + \langle s \rangle^{-\mu }\big) \ ds \notag \\
&\leq C|\xi|\int_0^t\delta^{\epsilon}\big(\langle s \rangle^{-(1-\epsilon/2  )\sigma} + \langle s \rangle^{-\mu}\big) \ ds.
\end{align}
In the third inequality above we have used (\ref{eqn:L2interp}). Assuming that $|\xi|^2 t \leq a$ and integrating we find that if $\mu \neq 1$ and $\mu \leq (1-\epsilon/2)\sigma$
then $|\hat{z}(t,\xi)| \leq C\delta^{\epsilon}|\xi|^{\rm min(2\mu - 1,
  1)}$ whereas if $\mu \geq (1-\epsilon/2)\sigma$ but
$(1-\epsilon/2)\sigma <1$ then $|\hat{z}(t,\xi)| \leq
C\delta^{\epsilon}|\xi|^{(2(1-\epsilon/2)\sigma - 1)}$.  If $\mu \geq
(1-\epsilon/2)\sigma>1$  then $|\hat{z}(t,\xi)| \leq
C\delta^{\epsilon}|\xi|$. Substituting
in the differential inequality (\ref{eqn:zdiffineq}) and integrating
we obtain (noting (\ref{eqn:zeesmalltime})) that in these three
circumstances $\|z(t)\| \leq C \langle t \rangle^{-\mu'/2}$ where
$\mu' = \rm min(2\mu +1/2, 5/2)$,  $\mu' = 1/2 +
2(1-\epsilon/2)\sigma$ or $\mu' = 5/2$, respectively. Tracing through the iterations starting from $\mu = 0$ we find that at most three iterations gives
 the stated result.
\end{proof}

\medskip

\noindent{\bf Acknowledgement}
I.H. would like to thank the Mathematics Institute of Aarhus University and the Schr\"{o}dinger Institute in Vienna for their support and hospitality while some of this research was completed.

%\begin{obs}\label{obs:global-stability}
%In the setting of Corollary \ref{cor:global-stability}, $M=I$ and $P$,
%one might
%suspect that $\vI_r=PH^r$ for all $r\geq 1/2$. Clearly this statement would be
%a consequence of closedness of the subset $\vI_r\subseteq PH^r$ for
%some $r\geq 1/2$.
 % \end{obs}

\end{document}